\documentclass[11pt,letterpaper]{article}

\newcommand{\defn}[1]{\textbf{\emph{#1}}}

\usepackage[T1]{fontenc}
\usepackage[margin=1in]{geometry}
\usepackage{amsfonts,amsmath,amsthm,amssymb,mathtools}  %

\mathtoolsset{centercolon}
\usepackage{xfrac,nicefrac}
\usepackage{mathdots}
\usepackage{mleftright}  %
\let\left\mleft
\let\right\mright

\usepackage{xspace}
\xspaceaddexceptions{]\}}  %
\usepackage{regexpatch}

\usepackage{bm,bbm,dsfont}  %
\usepackage{caption}
\usepackage[normalem]{ulem}
\usepackage{enumitem}

\usepackage{graphicx}
\usepackage{float}
\usepackage{subcaption}  %
\usepackage{tcolorbox}
\usepackage{tikz}
\usetikzlibrary{decorations.pathreplacing}
\usetikzlibrary{calc}
\usetikzlibrary{positioning}
\usetikzlibrary{arrows.meta}

\usepackage[linesnumbered,boxed,ruled,vlined]{algorithm2e}

\SetCommentSty{mycommfont}

\usepackage{thmtools,thm-restate}
\usepackage[colorlinks,citecolor=blue,linkcolor=blue,urlcolor=red]{hyperref}
\usepackage{breakurl}

\theoremstyle{plain}

\newtheorem{theorem}{Theorem}[section]  %
\newtheorem{lemma}[theorem]{Lemma}

\newtheorem{proposition}[theorem]{Proposition}

\newtheorem{claim}[theorem]{Claim}

\theoremstyle{definition}  %
\newtheorem{definition}[theorem]{Definition}

\newenvironment{proofsketch}{\begin{proof}[Proof Sketch]}{\end{proof}}

\usepackage[capitalise]{cleveref}
\crefname{algocf}{Algorithm}{Algorithms}
\Crefname{algocf}{Algorithm}{Algorithms}
\crefname{claim}{Claim}{Claims}
\Crefname{claim}{Claim}{Claims}

\newcommand{\AlgorithmCaption}[3]{\refstepcounter{#1} \label{#3} \captionof*{#1}{#2}}
\newfloat{Distribution}{htbp}{loa}
\crefname{Distribution}{Distribution}{Distributions}
\Crefname{Distribution}{Distribution}{Distributions}
\newfloat{Protocol}{htbp}{loa}
\crefname{Protocol}{Protocol}{Protocols}
\Crefname{Protocol}{Protocol}{Protocols}

\SetKwProg{Function}{Function}{:}{}

\SetKwProg{CodeBlock}{}{}{}
\SetKwProg{Repeat}{Repeat}{:}{}

\DeclarePairedDelimiter{\ceil}{\lceil}{\rceil}
\DeclarePairedDelimiter{\floor}{\lfloor}{\rfloor}

\DeclarePairedDelimiter{\bk}{(}{)}
\DeclarePairedDelimiter{\Bk}{[}{]}
\DeclarePairedDelimiter{\BK}{\{}{\}}

\DeclarePairedDelimiter{\angbk}{\langle}{\rangle}
\DeclarePairedDelimiter{\abs}{\lvert}{\rvert}

\DeclarePairedDelimiterX\mysetbase[2]{\lbrace}{\rbrace}{#1\,\delimsize\vert\,#2}
\NewDocumentCommand{\myset}{sO{}m m}{%
  \IfBooleanTF{#1}%
    {\mysetbase*{#3}{#4}}%
    {\mysetbase[#2]{#3}{#4}}%
}

\DeclareMathOperator*{\E}{\mathbb{E}}

\let\Pr\PrAux
\DeclareMathOperator{\poly}{poly}

\DeclareMathOperator*{\ind}{\mathbbm{1}}

\renewcommand{\tilde}{\widetilde}

\newcommand{\defeq}{\coloneqq}
\newcommand{\eps}{\varepsilon}

\newcommand{\N}{\mathbb{N}}

\newcommand{\Z}{\mathbb{Z}}
\renewcommand{\l}{\ell}

\renewcommand{\epsilon}{\eps}

\newcommand{\numberthis}{\addtocounter{equation}{1}\tag{\theequation}}

\usepackage{regexpatch}
\makeatletter
\xpatchcmd\thmt@restatable{%
\csname #2\@xa\endcsname\ifx\@nx#1\@nx\else[{#1}]\fi
}{%
\ifthmt@thisistheone
\csname #2\@xa\endcsname\ifx\@nx#1\@nx\else[{#1}]\fi
\else
\csname #2\@xa\endcsname[{Restated}]
\fi}{}{}
\makeatother

\newcommand{\ProbeComplexity}{\textup{probe-complexity}}
\newcommand{\Parent}{\texttt{\textup{parent}}\xspace}
\newcommand{\Heavy}{\texttt{\textup{heavy}}\xspace}
\newcommand{\Light}{\texttt{\textup{light}}\xspace}
\newcommand{\Insert}{\textup{\textsc{Insert}}\xspace}
\newcommand{\Delete}{\textup{\textsc{Delete}}\xspace}
\newcommand{\Query}{\textup{\textsc{Query}}\xspace}
\newcommand{\Sample}{\textup{\textsc{Sample}}\xspace}
\newcommand{\Null}{\textup{\texttt{null}}\xspace}

\begin{document}

\title{Tight Bounds for Classical Open Addressing}
\author{
Michael A. Bender\thanks{Partially supported by NSF grants CCF 2247577 and  CCF 2106827 and  John L. Hennessy Chaired Professorship. \texttt{bender@cs.stonybrook.edu}.}\\
Stony Brook University and RelationalAI
\and
William Kuszmaul\thanks{Partially supported by a Harvard Rabin Postdoctoral Fellowship and by a Harvard FODSI fellowship under NSF grant DMS-2023528. \texttt{kuszmaul@cmu.edu}.}\\
CMU
\and
Renfei Zhou\thanks{\texttt{renfeiz@andrew.cmu.edu}.}\\
CMU
}
\date{}
\maketitle

\begin{abstract}
We introduce a classical open-addressed hash table, called \emph{rainbow hashing}, that supports a load factor of up to $1 - \epsilon$, while also supporting $O(1)$ expected-time queries, and $O(\log \log \epsilon^{-1})$ expected-time insertions and deletions. We further prove that this tradeoff curve is optimal: any classical open-addressed hash table that supports load factor $1 - \epsilon$ must incur $\Omega(\log \log \epsilon^{-1})$ expected time per operation.

Finally, we extend rainbow hashing to the setting where the hash table is \emph{dynamically resized} over time. Surprisingly, the addition of dynamic resizing does not come at any time cost---even while maintaining a load factor of $\ge 1 - \epsilon$ at all times, we can support $O(1)$ queries and $O(\log \log \epsilon^{-1})$ updates.

Prior to our work, achieving any time bounds of the form $o(\epsilon^{-1})$ for all of insertions, deletions, and queries simultaneously remained an open question.
\end{abstract}

\section{Introduction}

Open-addressing is a simple framework for hash-table design that captures many of the most widely-used hash tables in practice (e.g., linear probing, quadratic probing, double hashing, cuckoo hashing, graveyard hashing, Robin-Hood hashing, etc). What these hash tables share in common, and indeed, what makes them examples of open addressing, is that in each case:
\begin{enumerate}
    \item The data structure itself is just an array of some size $N$ containing elements, free slots, and (in some hash-table designs) tombstones.\footnote{In some hash tables, when an element is deleted, it is replaced with a \defn{tombstone}. Then, once there are sufficiently many tombstones in the hash table, they are all removed at once in a single batch, and the hash table is rebuilt. The constructions in this paper will not make use of tombstones, but we include hash tables that do in our discussions of prior work.}
    \item Each element $x$ has a \defn{probe sequence} $h_1(x), h_2(x), \ldots \in [N]$ that is fully determined by $x$, $N$, and random bits (i.e., hash functions).
    \item And the procedure for querying $x$ is to simply examine the array positions $h_1(x), h_2(x), \ldots$ either until $x$ is found or until the query is able to conclude, based on what it has seen, that $x$ is not present.
\end{enumerate}

In the decades since open addressing was first introduced, there have been dozens (or possibly even hundreds) of hash-table designs proposed within the open-addressing model. However, the most basic question that one could have remained open: \emph{What is the best space-time tradeoff that any open-addressed hash table can achieve?} 

\paragraph{The space vs.~time tradeoff. } A hash table is said to support a \defn{load factor} of $1 - \epsilon$ if it supports sequences of insertions/deletions/queries with up to $\lceil (1 - \epsilon)N \rceil$ elements present at a time. As $\epsilon$ decreases, the space efficiency improves, but the time per operation gets worse. The \defn{space-time tradeoff} for a hash table is the relationship between $\epsilon$ and the expected time for each type of operation (insertions, deletions, and membership queries).

In some hash-table constructions, the hash table supports \defn{dynamic resizing}, meaning that the hash table changes the size $N$ of the array that it uses, over time, in order to maintain a load factor of at least $1 - \epsilon$ at all times. Hash tables that do not do this (i.e., that use a fixed $N$) are referred to as having a \defn{fixed-capacity}. We will be interested in both fixed-capacity and dynamically-resized hash tables in this paper.

\paragraph{The historical barrier: Achieving $o(\epsilon^{-1})$-time operations.} It is relatively straightforward to construct an open-addressed hash table that achieves $O(\epsilon^{-1})$ expected-time operations (e.g., by using uniform probing with tombstones \cite{Knuth98Vol3}). With more sophisticated techniques \cite{brent1973reducing, gonnet1979efficient, mallach1977scatter, fotakis2005space, dietzfelbinger2007balanced}, one can achieve expected query time $o(\epsilon^{-1})$ while supporting insertions and deletions in $f(\epsilon) = \Omega(\epsilon^{-1})$ time for some $f$. In fact, if all that one cares about are \emph{positive} queries (and if one disregards insertion/deletion time entirely), then the expected query time can even be reduced to $O(1)$ \cite{brent1973reducing, gonnet1979efficient, mallach1977scatter}.

It has remained an open question whether one might be able to achieve an expected time bound of $o(\epsilon^{-1})$ for all operations simultaneously. This is not to say that there is no hope. It is \emph{conjectured}, for example, that bucketed cuckoo hashing \cite{dietzfelbinger2007balanced} with buckets of size $\Theta(\sqrt{\epsilon^{-1}})$ (and using random-walk insertions) can support $\tilde{O}(\sqrt{\epsilon^{-1}})$-time operations---but even bounding the insertion time by $f(\epsilon^{-1})$ for any function $f$ remains an open question \cite{dietzfelbinger2007balanced, bucketcuckoorandom}. If one further brings down the target time to $o(\sqrt{\epsilon^{-1}})$ per operation, then there are not even any conjectured solutions. All known solutions that achieve $O(t)$-time queries for some $t \le \epsilon^{-1}$ require insertions to spend at least $\Omega(\epsilon^{-1}/t)$ time even just deciding which free slot to consume \cite{brent1973reducing, gonnet1979efficient, mallach1977scatter, fotakis2005space, dietzfelbinger2007balanced}. 

It is worth remarking that there is an additional bottleneck if one wishes to support dynamic resizing. The standard approach to maintaining a dynamic load factor of $\ge 1 - \epsilon$ is to rebuild the hash table every time that $\Theta(\epsilon n)$ insertions or deletions occur. These rebuilds require $\Omega(n)$ time (even just to read through the entire data structure), thereby contributing at least $\epsilon^{-1}$ amortized expected cost per insertion/deletion. Even if one could achieve $o(\epsilon^{-1})$-time operations for a fixed-capacity table, it is not clear whether such a guarantee could be extended to the dynamic-resizing case.

\paragraph{This paper: Optimal open-addressing.} In this paper, we introduce \defn{rainbow hashing}, an open-addressed hash table that achieves expected query time $O(1)$ and expected insertion/deletion time $O(\log \log \epsilon^{-1})$. The name of the data structure refers to the way in which it assigns colors to elements in order to decide the layout of the hash table. 

Our second result is an extension of rainbow hashing that supports dynamic resizing without changing the time bounds. That is, even while maintaining the invariant the $N \ge (1 - \epsilon) n$ at all times, the hash table is able to support $O(1)$ expected-time queries and $O(\log \log \epsilon^{-1})$ expected-time insertions/deletions.

An interesting consequence is what happens at a load factor of $1$. Here, our construction yields an open-addressed hash table with constant expected-time queries, with $O(\log \log n)$ expected-time insertions/deletions, and where the hash table is fully compacted at all times---that is, the number of slots $N$ that the hash table uses is always the \emph{same} as the number of elements $n$ that it contains. 

Finally, we conclude the paper by proving a matching lower bound: we show that, in any open-addressed hash table that supports positive queries in $O(\log \log \epsilon^{-1})$ expected time, the amortized expected time per insertion/deletion must be at least $\Omega(\log \log \epsilon^{-1})$. At a technical level, our lower bound can be viewed as a (highly nontrivial) extension of the potential-function techniques previously developed in \cite{benderhashing} for analyzing the average log-probe-complexity in a hash table. 

Interestingly, the lower bound applies not just to insertion/deletion \emph{time} but also to the \emph{number of items} that the insertion/deletion rearranges. Thus the result applies even to hash tables that go beyond `pure' open addressing, and that store arbitrary amounts of additional metadata to help with the operation of the data structure.

Combined, our results fully resolve the optimal time complexity of open addressing. 

\paragraph{Other related work.} It is not possible to describe the entire body of work on open addressing, so we will focus instead on summarizing the high-level trajectory of the area. Early work \cite{Knuth63, Peterson57, Knuth98Vol3, AmbleKn74, CelisLaMu85, KonheimWe66}, spanning the late 1950s through the early 1970s, focused largely on evaluating variations of linear probing \cite{Knuth63, Peterson57, KonheimWe66, KonheimWe66}, quadratic probing \cite{Maurer68}, and uniform probing \cite{Peterson57, Knuth98Vol3}. In a significant 1973 breakthrough \cite{brent1973reducing}, Brent showed that one could support $O(1)$-expected-time positive queries regardless of $\epsilon$. Brent's result prompted a large body of work on query-optimized hash tables. This included both lower bounds \cite{yao1985uniform, ajtai1978there} on restricted classes of hash tables, alternative techniques for obtaining fast queries \cite{gonnet1979efficient, mallach1977scatter}, and in the past few decades, the emergence of \emph{cuckoo hashing} \cite{fotakis2005space, dietzfelbinger2007balanced, pagh2001cuckoo}. Cuckoo hashing, when implemented with the appropriate parameters \cite{fotakis2005space, dietzfelbinger2007balanced}, can be used to support $O(\log \epsilon^{-1})$-time queries, not just in expectation, but even in the \emph{worst case}. Although this time bound is worse than the $O(1)$-bound achieved by Brent and others \cite{brent1973reducing, gonnet1979efficient, mallach1977scatter}, it is notable for applying to both positive \emph{and} negative queries. 

A common feature among many open-addressing schemes is that, even in cases where the hash table's behavior is intuitively easy to understand, it is often quite difficult to analyze. The variants of cuckoo hashing that achieve $O(\log \epsilon^{-1})$-time queries, for example, are conjectured to also achieve $O(\epsilon^{-1})$-time insertions \cite{dietzfelbinger2007balanced, fotakis2005space, bucketcuckoorandom, bell20241}, but even bounds of the form $\poly(\epsilon^{-1})$ remain open. Another example is quadratic probing \cite{Maurer68, HopgoodDa72}, which despite widespread use since the early 1970s \cite{WikipediaQuadraticProbing}, has resisted time bounds of the form $f(\epsilon^{-1})$ for \emph{any} $f$. Other examples, still, include double hashing, the analysis of which remained open for decades \cite{GuibasSz78, lueker1988more}, and robin-hood hashing \cite{CelisLaMu85} (a.k.a.~ordered linear probing \cite{AmbleKn74}), which was only very recently revealed to yield much better time bounds than were previously thought to hold \cite{bender2022linear}. 

In recent decades, there has also been a great deal of work on hash tables that go beyond the classical open-addressing model. This includes work on succinct data structures \cite{benderhashing, li2023dynamic, li2023tight, bender2023iceberg, bercea2023dynamic, Raman03Succinct, ArbitmanNaSe10}, hash tables with high-probability time guarantees \cite{GoodrichHiMi12, kuszmaul2022hash, bender2023iceberg}, and dictionaries in the external-memory model \cite{IaconoPa12, ConwayFaSh18, JensenPa08, verbin2013limits}. In the area of succinct data structures, in particular, there have been a series of recent breakthroughs \cite{benderhashing, li2023dynamic, li2023tight} that, together, fully characterize the optimal time-space tradeoff curve of any unordered dictionary. Interestingly, many of the data structures introduced in this line of work share a subtle connection to open addressing: they can be viewed as classical open-addressed hash tables in which the probe position of each item $x$ is cached in a secondary data structure. Roughly speaking, this allows one to take an open-addressed hash table in which a key $x$ would have required $t$ time to find, and to instead store $\log t$-bit number $t$ explicitly (in the secondary data structure), so that $x$ can be found in constant time. This connection led Bender et al.~\cite{benderhashing} to study a variation of open addressing in which the goal is to minimize the average \emph{log} query time over all elements. Not surprisingly, this leads to a very different tradeoff curve than the one in this paper---for example, $O(1)$ log query time corresponds to $O(\log^* n)$-time insertions. Nonetheless, as we show in Section \ref{sec:lower}, the lower-bound techniques from \cite{benderhashing} can actually be extended to our setting, albeit, with several significant changes.

\section{Preliminaries}

\paragraph{Open addressing. }At a high level, an \defn{open-addressed hash table} with \defn{capacity} $N$ is a hash table that stores its keys (elements of some universe $U$) in an array of size $N$, and that uses a \defn{probe-sequence function} $h(x) = \langle h_1(x), h_2(x), \ldots \rangle \in [N]^{\infty}$ to perform queries. 

The \defn{state} of the hash table is an array $A$ of $N$ slots, where each slot either stores a key or is a free slot. Specifically, if $S \subseteq U$ is the current set of keys, then each key $x \in S$ appears exactly once in $A$, and the remaining $N - |S|$ slots are left empty. 

In addition to its state $A$, the hash table gets to store the capacity $N$ of the array and the current number $n = |S|$ of elements. The hash table can also invoke a fully random hash function which it is not responsible for storing. These are the only things that the hash table has access to.

Queries $\Query(x)$ for a key $x$ are implemented by scanning a probe sequence $h(x) = \angbk{h_1(x), h_2(x), \ldots}$ of positions in the array, until either $x$ is found or until some stopping condition is met. The probe sequence for $x$ must depend on only $x$, $N$, and random bits (i.e., hash functions). Insertions and deletions are permitted to rearrange the elements in the hash table however they wish, so long as they preserve the correctness of queries. 

The \defn{load factor} of the hash table is defined to be $n / N$, and is typically denoted by $1 - \epsilon$. The goal is to design insertion, deletion, and query algorithms that allow for time-efficient operations as a function of $\epsilon^{-1}$.

\paragraph{Fixed capacity vs.~dynamic resizing.} A classical open-addressing hash table is said to have \defn{fixed capacity} if $N$ remains the same over time, and is said to be \defn{dynamically resized} if $N$ changes over time (so that, at any given moment, the hash table resides in the first $N$ slots of an infinite array). We will be interested in proving upper bounds for both the fixed-capacity and dynamically-resized cases. Our lower bounds will be for the fixed-capacity case but using inputs in which the total number of elements changes by only $\pm 1$ over time.

When discussing fixed-capacity hash tables, we will aim to support a large maximum load factor. When discussing dynamically-resized hash tables, we will aim to support a continual load factor---that is, as $n$ changes, $N$ will also change to preserve the load factor. 

\paragraph{Other variations.} The flavor of open addressing studied in this paper is the most classical version of the problem---the one typically referred to simply as \emph{open addressing} \cite{munro1986techniques, Knuth98Vol3}. However, there are also many other variations that have been studied, some of which have also been referred to as open addressing: notable examples include \emph{greedy} open addressing, where each insertion must use the first free slot it finds \cite{ullman1972note, yao1985uniform}; and \emph{non-oblivious} open addressing, where the querier can probe the hash table adaptively rather than using a fixed probe sequence \cite{fiat1988nonoblivious, fiat1993implicit}. Thus the term \emph{open addressing} is not always unambiguous in the literature, and to emphasize the fact that we are focusing on standard open addressing (rather than the greedy or non-oblivious variations), we will refer to the class of hash tables that we are studying as \defn{classical open addressing} for the rest of the paper.

\section{The Rainbow Cell}\label{sec:cell}

In this section, we describe a simple hash table called the \defn{rainbow cell}. The rainbow cell operates continuously at a load factor of 1. That is, the hash table is initialized to contain $n$ elements in $n$ slots, and then each update to the hash table both deletes some element and inserts some new element, so that every slot is still occupied. One could also extend the rainbow cell to support load factors less than 1, but as we shall see, this will not actually be necessary for our applications of it.

What makes the rainbow cell a bit unusual, as a hash table, is that it prioritizes update time over query time. Indeed, whereas updates will take expected time $O(1)$, queries will be permitted to take expected time $O(n^{3/4})$. This may seem counterproductive, given that our final data structure, the rainbow hash table (Section \ref{sec:rainbowhashing}), will support constant-time \emph{queries} and (slightly) super-constant-time insertions/deletions. Nonetheless, the rainbow cell will end up serving as a critical building block in constructing the full rainbow hash table.

Finally, for this section, we will assume that deletions \emph{already know the position of the element being deleted}. In many hash tables, including those constructed in later sections, this assumption is without loss of generality, because the expected time to query an element is less than the intended expected deletion time. The assumption is \emph{not} without loss of generality for the rainbow cell, because queries are slower than deletions, but the assumption will turn out to be true (by design) in our applications of rainbow cells later on.

In the rest of the section, we will define how a rainbow cell works and prove the following proposition: 

\begin{proposition}
The rainbow cell is a classical open-addressed hash table that operates continuously at load factor $1$, supports updates in expected time $O(1)$, and supports queries in expected time $O(n^{3/4})$. Updates consist of a deletion and insertion pair, and assume that the element being deleted is in a position already known (i.e., that we do not need to perform a query to find the element).
\label{prop:cell}
\end{proposition}

\paragraph{The data structure.} The rainbow-cell partitions an array into $n^{1/4}$ \defn{buckets}, each of which has $n^{3/4}$ slots. The final $n^{1/2}$ slots in each bucket will play a special role, and are referred to as \defn{sky slots}. Across all $n^{1/4}$ buckets, there are $n^{3/4}$ total sky slots.

Each element $x$ will be assigned a status as either \defn{heavy} or \defn{light}. The status is given to $x$ by a \defn{status hash function} $s(x)$ that sets $s(x) = \Heavy$ with probability $1 - 1 / (2n^{1/4})$ and $s(x) = \Light$ with probability $1 / (2n^{1/4})$. If an element $x$ is heavy, then it is also assigned a random bucket $h(x)$. 

Whenever possible, the state of the data structure will be as follows: each heavy element $x$ is stored in its assigned bucket $h(x)$, and each light element $y$ is stored in some sky slot of some bucket (any bucket will do). If the elements are stored in this type of configuration, we say the hash table is in a \defn{common-case configuration}. An example is given in Figure \ref{fig:cell}.

\begin{figure}
    \centering
    \includegraphics[height = 6 cm]{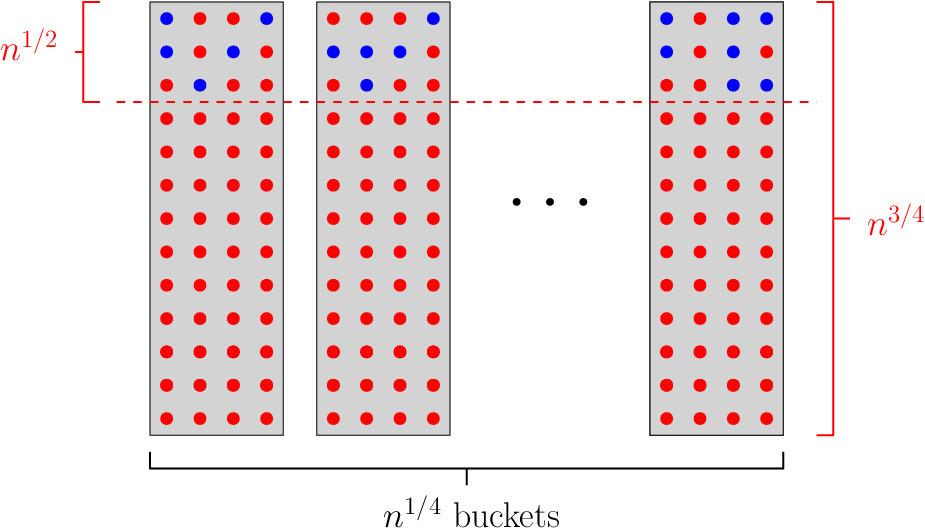}
    \caption{A rainbow cell in a common-case configuration. The red elements are heavy elements, each of which is in their assigned bucket; the blue elements are light elements, each of which is in a sky slot of some bucket.}
    \label{fig:cell}
\end{figure}

Assuming a common-case configuration, queries are straightforward to implement: Heavy elements $x$ are queried by scanning the single bucket $h(x)$, and light elements can be queried by scanning the sky slots of all buckets. Both cases result in $O(n^{3/4})$-time queries.

Update operations will keep the data structure in a common-case configuration whenever it is possible to do so. If it is not possible, because some bucket has either fewer than $n^{3/4} - n^{1/2}$ or greater than $n^{3/4}$ (heavy) elements that hash to it, then the hash table is said to have incurred a \defn{full failure}. Whether or not the hash table is incurring a full failure is indicated to queries by the relative order of the final two elements in each bucket (i.e., whether the final two elements $a$ and $b$ satisfy $a < b$ or $b < a$). If a query observes that a full failure has occurred, then after it has finished the $n^{3/4}$ probes that it would normally perform, it performs $n$ additional probes to check the entire hash table for the element being queried. Thus the query time is $O(n^{3/4})$ whenever a full failure has \emph{not} occurred, and is $O(n)$ whenever a full failure has occurred.

Finally, we must describe how to implement updates so that (1) the hash table is in a common-case configuration whenever possible; (2) the update detects whenever a full failure occurs; and (3) the expected time for the update is $O(1)$. 

To describe the update operation, it is helpful to first assume what we call a \defn{update-friendly input}: such an input is one in which every bucket $j \in [1, n^{1/4}]$ has between $n^{3/4} - \frac{2}{3} \sqrt{n}$ and $n^{3/4} - \frac{1}{3} \sqrt{n}$ heavy items that hash to it. What is nice about update-friendly inputs is that the sky slots in each bucket are are guaranteed to be at least $1/3$ heavy items and at least $1/3$ light items. This means that if we wish to find a light (resp.~heavy) item within the sky slots of some bin, we can do so in $O(1)$ expected time by simply sampling random sky slots in the bin until we find a light (resp.~heavy) item. We call this the \defn{sampling trick}.

Assuming an update-friendly input, we can implement updates as follows. Suppose we wish to delete some item $x$, currently in some bucket $j$, and insert some item $y$. The update can be performed as follows: 
\begin{itemize}
    \item We begin by removing $x$, which creates a free slot $s$ in bucket $j$. 
    \item If $s$ is not a sky slot, then we find some heavy element $x'$ that \emph{is} in a sky slot of bucket $j$ (we can do this in $O(1)$ expected time using the sampling trick). We then move $x'$ to slot $s$, freeing up some sky slot $s'$. For notational convenience, if $s$ was already a sky slot, then we simply define $s' = s$. Thus, at the end of this step, we have created a free sky slot $s'$ in bin $j$.
    \item If $y$ is a light element, then we complete the update by placing it in slot $s'$. If $y$ is a heavy element, then we need to create a free slot $s''$ in bin $h(y)$ and place $y$ there. To do this, we find a light element $z$ in some sky slot $s''$ of bin $h(y)$ (again, this step takes $O(1)$ expected time using the sampling trick).  We then move $z$ to slot $s'$, and place $y$ in slot $s''$. This completes the update, while keeping the hash table in a common-case configuration.
\end{itemize}

Whenever the hash table is in a common-case configuration, each update will attempt to use the above protocol. If the protocol succeeds, then the hash table will continue to be in a common-case configuration, as desired. The protocol is said to \defn{fail} if it runs for time $n$ without completing. In this case, the update swaps to a more naive method: it scans the entire hash table, determines whether we are experiencing a full failure, and rebuilds the entire hash table appropriately. This failure mode guarantees that the update takes worst-case $O(n)$ time.

Finally, whenever the hash table is \emph{not} in a common-case configuration (i.e., it is already experiencing a full failure), the updates default to the $O(n)$-time failure mode: they scan the entire hash table and rebuild it in whatever way is appropriate based on whether the hash table is still experiencing a full failure or not.

\paragraph{Analyzing the rainbow cell.} Having described how the rainbow cell works, the analysis follows from a straightforward application of Chernoff bounds.

\begin{lemma}
\label{lem:insertionfriendly}
Each input is update-friendly with probability $1 - n^{-\omega(1)}$.
\end{lemma}
\begin{proof}
The expected number of heavy elements that hash to a given bucket is $n^{3/4} - \frac{1}{2} n^{1/2}$. It follows by a Chernoff bound that, with probability $1 - n^{-\omega(1)}$, the number of such elements will be between  $n^{3/4} - \frac{2}{3} \sqrt{n}$ and $n^{3/4} - \frac{1}{3} \sqrt{n}$. Moreover, by a union bound, the probability of this bound failing for any bucket $j$ is at most $n^{3/4} \cdot n^{-\omega(1)} = n^{-\omega(1)}$. 
\end{proof}

\begin{lemma}
\label{lem:cellupdate}
The expected time to perform an update is $O(1)$.
\end{lemma}
\begin{proof}
If the input is update-friendly, then the sampling trick allows us to complete the update in $O(1)$ expected time. If the input is not update-friendly, then the update may take time as much as $O(n)$. Thus, by Lemma \ref{lem:insertionfriendly}, the expected update time is at most
\[O(1) + \Pr[\text{non-update-friendly input}] \cdot O(n) = O(1) + n^{-\omega(1)} \cdot O(n) = O(1). \qedhere\]
\end{proof}

\begin{lemma}
The expected time to perform a query is $O(n^{3/4})$.
\label{lem:cellquery}
\end{lemma}
\begin{proof}
If the hash table is in a common-case configuration, then the query takes time $O(n^{3/4})$. If the hash table is in a non-common-case configuration, then it is experiencing a full failure, and the query may take time as much as $O(n)$. The probability of a full failure occurring is at most the probability that the input is non-update-friendly. Thus, by Lemma \ref{lem:insertionfriendly}, the expected update time is at most
\[O(n^{3/4}) + \Pr[\text{non-update-friendly input}] \cdot O(n) = O(n^{3/4}) + n^{-\omega(1)} \cdot O(n) = O(n^{3/4}).\qedhere\]
\end{proof}

Since the rainbow cell is, by design, a classical open-addressed hash table that operates at load factor $1$, we can combine Lemmas \ref{lem:cellupdate} and \ref{lem:cellquery} to obtain a proof of Proposition \ref{prop:cell}.

\section{The Rainbow Hash Table}\label{sec:rainbowhashing}

In this section, we describe the most basic version of \defn{rainbow hashing}. For now, we will confine ourselves to the setting where the hash table operates continually at a load factor of 1 (we will support other load factors later on in Section \ref{sec:loadfactor}), and where the hash table does not resize (we will add dynamic resizing in Section \ref{sec:resizing}). Subject to these restrictions, we will achieve $O(1)$ expected-time queries and $O(\log \log n)$ expected-time updates.

\begin{proposition}
Basic rainbow hashing is a classical open-addressed hash table that operates continually at load factor $1$, achieves $O(1)$ expected-time queries, and achieves $O(\log \log n)$ expected-time updates.
\label{prop:rainbow}
\end{proposition}

\paragraph{Defining a recursive structure. }
To describe basic rainbow hashing, we will think of the array as being broken into the recursive tree structure shown in Figure \ref{fig:rainbowrecursion}. The root of the tree, which occurs at some level $\overline{\ell} = \Theta(\log \log n)$, is a single subproblem that contains the entire array---its size is denoted by $n_{\overline{\ell}} \defeq n$. The leaves of the tree, which occur at level $1$, are subproblems that consist of a buffer of $b_1 = n_{1} = O(1)$ elements. For every level $i \in (1, \overline{\ell}]$, the subproblems in level $i$ have the following structure: each level-$i$ subproblem has size $n_i$, and it consists of a \defn{level-$i$ buffer} of size $b_i = n_i^{0.51}$ as well as $f_i = n_i^{0.01}$ children that are each level-$(i - 1)$-nodes of sizes $$n_{i - 1} := (n_i - b_i) / f_i = \Theta(n_i^{0.99}).$$
Note that the recursive relationship between $n_i$ and $n_{i - 1}$ determines all of the values $\{n_i\}, \{b_i\}, \{f_i\}$ in the tree. Finally, it will be helpful to also define $m_1, m_2, \ldots, m_{\overline{\ell}}$ so that $m_i = \Theta(n / n_i)$ is the number of level-$i$ subproblems.

Each item $x$ will hash to a random \defn{color} $C(x) \in \{2, 3, \ldots, \overline{\ell}\}$ (we will specify the distribution on $C(x)$ in a moment), and to a uniformly random subproblem $h(x)$ out of the $m_{C(x)}$ subproblems with level $C(x)$. The probability distribution for $C(x)$ is given by $\Pr[C(x) = i] = p_i$ for a set of values $p_2, p_3, \ldots, p_{\overline{\ell}}$ satisfying $\sum_i p_i = 1$ and such that, for $1 < i < \overline{\ell}$, 
\begin{equation}
(p_2 + \cdots + p_i)n \in \left(\sum_{j = 1}^{i} m_j \cdot b_j\right) - [0.4, 0.6] \cdot m_i \cdot b_i.
\label{eq:pi}
\end{equation}
It may be helpful to think about this inequality at a per-subproblem level: It is equivalent to say that, for a subproblem $s$ in level $1 < i < \overline{\ell}$, if we define $t_s$ to be the sum of the sizes of the buffers of $s$ and all of $s$'s descendants, and we define $q_s$ to be the total number of elements that hash to $s$ and $s$'s descendants, then 
\begin{equation}\E[q_s] = t_s - [0.4, 0.6] \cdot b_i.
\label{eq:corereq}
\end{equation}
The use of the range $[0.4, 0.6]$ in \eqref{eq:pi} and \eqref{eq:corereq} will not be important in this section (we could just as well use the concrete value $0.5$), but it will be important later on in Section \ref{sec:resizing} where, in order to support resizing, we will allow $m_i$ to change over time.

\begin{figure}
    \centering
    \includegraphics[height = 7 cm]{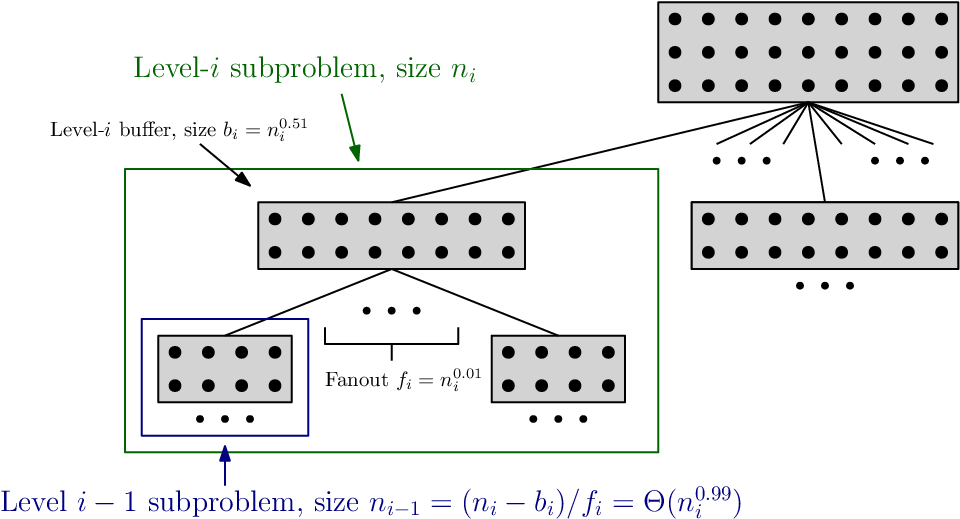}
    \caption{The recursive structure of a rainbow hash table.}
    \label{fig:rainbowrecursion}
\end{figure}

We conclude the definition of the $\{p_i\}$'s by noting their asymptotic relationship to $n_i, b_i$:
\begin{lemma}
For $i > 1$, each $p_i$ satisfies $p_i = \Theta(b_{i - 1} / n_{i - 1})$.
\label{lem:pi}
\end{lemma}
\begin{proof} For $i = 2$ (or $i = O(1)$), the lemma is trivial. By \eqref{eq:pi}, we have for $2 < i < \overline{\ell}$  that
    \begin{align*} 
    p_i n & \in m_i b_i - [0.4, 0.6] \cdot m_i b_i + [0.4, 0.6] \cdot m_{i - 1} b_{i - 1},
    \end{align*}
    and that for $i = \overline{\ell}$,
    \begin{align*} 
    p_i n & \in m_i b_i + [0.4, 0.6] \cdot m_{i - 1} b_{i - 1}.
    \end{align*}
    Either way,
    \begin{align*}
    p_i n & = \Theta(m_i b_i) + \Theta(m_{i - 1} b_{i - 1}) \\
    & = \Theta(m_{i - 1} b_{i - 1}),
    \end{align*}
    where the final step uses the fact that the buffers in level $i - 1$ have a larger cumulative size than those in level $i$.
    It follows that
    \[
    p_i = \Theta\left(\frac{b_{i - 1}}{n/m_{i - 1}}\right) = \Theta(b_{i - 1} / n_{i - 1}). \qedhere
    \]
\end{proof}

\paragraph{The role of rainbow cells.}
For each subproblem $s$ in the tree, we implement $s$'s buffer as a rainbow cell. This means that updates to the buffer can be implemented in $O(1)$ expected time and that queries can be implemented in $O(b_i^{3/4})$ expected time. The fact that queries to $s$'s buffer take time polynomially smaller than $b_i$ will end up being critical to the data structure.

\paragraph{Separating insertions and deletions. }So that we can discuss insertions and deletions separately, we will allow the hash table to take intermediate states containing a free slot in some (known) position. Deletions create such a free slot, and guarantee that the free slot appears in the buffer of the root subproblem; and insertions consume such a free slot, assuming that it was initially in the buffer of the root subproblem. Thus, several of the definitions that follow (namely the definitions of common-case configurations and of the Common-Case Invariant) should be viewed as applying both in the context where there is a free slot somewhere in the hash table and in the context where there is not.

\paragraph{High-level overview: the common-case behavior of the data structure.}
To motivate the data structure, let us take a moment to describe how the data structure will behave when certain rare (at a per-subproblem level) failure events do not occur. Precluding such failure events, the data structure will look as in Figure \ref{fig:rainbowcolors}: for each level-$i$ subproblem $s$, a constant fraction of the elements in $s$'s buffer will be color-$i$ elements that hash to $s$ (unless $s$ is a leaf), and the remaining will be color-$(i + 1)$ elements that hash to $s$'s parent (unless $s$ is the root). Notably, all of the elements in the table that hash to $s$ will be stored in either $s$'s buffer or in $s$'s children's buffers. 

\begin{figure}
    \centering
    \includegraphics[height = 6 cm]{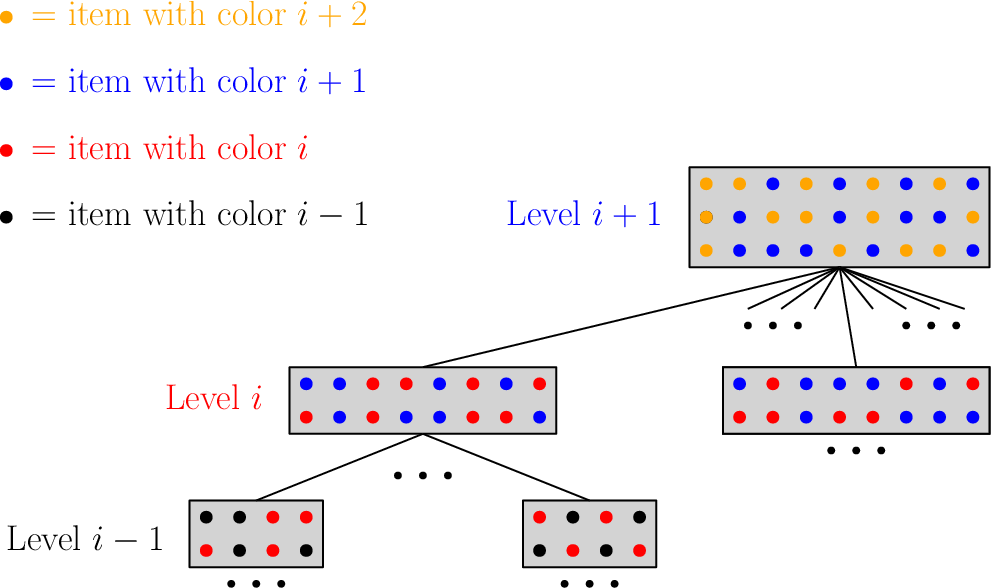}
    \caption{For $i < \overline{\ell}$, the common-case state of a level-$i$ node $s$ will be that a constant fraction of the elements in its buffer have color $i$ (and hash to $s$) and the rest of have color $i + 1$ (and hash to $s$'s parent).}
    \label{fig:rainbowcolors}
\end{figure}

It follows that, to query an element $x$ that hashes to $h(x) = s$, we can simply query $s$'s buffer and $s$'s children's buffers. Since each of these buffers is implemented as rainbow cells, this takes expected time
\[O\left(b_i^{3/4} + f_i \cdot b_{i - 1}^{3/4}\right).\]
This may seem large, but remember that most elements $x$ hash to very low-level nodes $s$. By Lemma \ref{lem:pi}, for $i > 1$, the probability that $x$ hashes to a level-$i$ node is $p_i = \Theta(b_{i - 1} / n_{i - 1})$ , so the expected time to query an element is

\begin{align*}
&    O\left(1 + \sum_{i > 1} \frac{b_{i - 1}}{n_{i - 1}} \cdot \left (b_{i}^{3/4} + f_{i} \cdot b_{i - 1}^{3/4}\right)\right) \\ 
& =  O\left(1 + \sum_{i > 1} \frac{n_i^{0.99 \cdot 0.51}}{n_i^{0.99}} \cdot \left ((n_i^{0.51})^{3/4} + n_i^{0.01} \cdot ((n_{i}^{0.99})^{0.51})^{3/4}\right)\right) \\
& =  O\left(1 + \sum_{i > 1} \frac{n_i^{0.99 \cdot 0.51}}{n_i} \cdot \left (n_i^{0.01} \cdot (n_i^{0.51})^{3/4}\right)\right) \\
& =  O\left(1 + \sum_{i > 1} n_i^{0.99 \cdot 0.51 - 1 + 0.01 + 0.51 \cdot 3/4}\right) \\
& =  O\left(1 + \sum_{i > 1} n_i^{-0.1}\right) \\
& = O(1).
 \end{align*}

Again, this is all assuming that certain rare failure events do not occur, but we will come back to that later.

To implement insertions and deletions, an important insight is that we can make use of a ``sampling trick'' similar to the one we used in rainbow cells. Suppose that we are looking at a subproblem $s$ (that is not a leaf) and that we wish to find an element $y$ in $s$'s buffer that hashes specifically to $s$. Since a constant fraction of the elements in $s$'s buffer hash to $s$, we can use random sampling to find such an element in $O(1)$ expected time. Likewise, if $s$ is not the root, and if we wish to find an element in $s$ that hashes to $s$'s parent, we can also do this in $O(1)$ expected time using random sampling.

As we perform the insertion/deletion, we need to preserve the invariant that, for each subproblem $s$,  all of the elements that hash to $s$ are stored in either $s$'s buffer or in $s$'s children's buffers (again, this ignores certain rare failure events). How can we implement insertions and deletions while preserving this invariant? 

Suppose we wish to insert an element $x$ that hashes to some subproblem $h(x)$. We want to place $x$ in the buffer of $h(x)$, but we currently have a free slot in the buffer of the root subproblem $r$. Let $s_1 = h(x), s_2, s_3, \ldots, s_j = r$ be the path from subproblem $h(x)$ to the root subproblem $r$. We can use the sampling trick to find elements $y_1, y_2, \ldots, y_{j - 1}$ in the buffers of $s_1, s_2, \ldots, s_{j -1}$ with the properties that $h(y_i) = s_{i + 1}$. We can then place $y_{j - 1}$ in the free slot in the root's buffer, place $y_{j - 2}$ in $y_{j - 1}$'s former position, place $y_{j - 3}$ in $y_{j - 2}$'s former position, and so on, ultimately freeing up the position where $y_1$ resided. Finally, we can put $x$ in this position. With this augmenting-path approach, we can complete the insertion while preserving the invariant that every item appears in the buffers of either the subproblem it hashes to or one of that subproblem's children.

Likewise, suppose we wish to delete an element $x$ that is currently in the buffer of some subproblem $s$. By removing $x$, we create a free slot in $s$'s buffer that we need to move to the root. Let $s_1 = s, s_2, s_3, \ldots, s_j = r$ be the path from subproblem $h(x)$ to the root subproblem $r$. We can use the sampling trick to find elements $y_2, y_3, \ldots, y_{j}$ in the buffers of $s_2, s_2, \ldots, s_{j}$ with the properties that $h(y_i) = s_{i}$. We can then place $y_{2}$ in the free slot in $s$'s buffer, place $y_{3}$ in $y_{2}$'s former position, place $y_{4}$ in $y_{3}$'s former position, and so on, ultimately freeing up a slot in $r$'s buffer. Critically, we have once again preserved the invariant that every item appears in the buffers of either the subproblem it hashes to or one of that subproblem's children.

This completes the description of how the data structure would behave if certain rare (at a per-subproblem level) failure events never occurred. The expected time per query would be $O(1)$ and the expected time per update would be $O(\log \log n)$. In general, however, we must be able to handle failure events that break our desired invariants (i.e., that cause the population in the buffer of a node to not simply be a constant fraction of elements that hash to the node and a constant fraction of elements that hash to the parent). Most of the effort in formalizing the data structure will be in designing and maintaining an invariant (that we will call the Common-Case Invariant) that lets us handle these failure events cleanly. We now continue in the rest of the section by presenting and analyzing the full data structure.

\paragraph{Storing a boolean in each buffer. }
We shall assume for the sake of discussion that, for each subproblem $s$, we can store a boolean (called the \defn{failure indicator}) indicating whether $s$ is in a certain type of failure mode (to be specified later). This boolean will not affect the probe sequence that queries perform (since, after all, queries must be oblivious), but will help queries determine when they can terminate. We emphasize that the failure indicator can easily be encoded implicitly in the relative order of, say, the first two elements of $s$. Thus it is only for ease of discussion that we treat the boolean as being stored explicitly.

\paragraph{The Common-Case Invariant. } 
A subproblem $s$ is said to be in a \defn{weakly common-case configuration} if:
\begin{itemize}
\item For each child $c$ of $s$, all of the elements that hash to $c$ and its descendants are stored in the buffers of $c$ and its descendants. 
\item The only elements in $s$'s descendants' buffers that do not hash to $s$'s descendants are elements that hash to $s$ and are in $s$'s children's buffers.
\end{itemize}
The subproblem $s$ is further said to be in a \defn{strongly common-case configuration} if both $s$ and $s$'s ancestors are all in weakly common-case configurations. Critically, this implies that all of the elements that hash to $s$ are in the buffers of $s$ and $s$'s children, and that all of the elements in $s$'s buffer that do not hash to $s$ are elements that hash to $s$'s parent. 

A subproblem $s$ is said to be \defn{weakly feasible} if it is possible to arrange the elements in the tree so that $s$ is in a weakly common-case configuration, \emph{and} so that any free slot in the tree (if we are in an intermediate state between a deletion and an insertion) is not contained in the buffers of $s$'s descendants. The subproblem $s$ is said to be \defn{strongly feasible} if all of $s$ and its ancestors are weakly feasible.

Our algorithm will maintain what we call the \defn{Common-Case Invariant:}
\begin{itemize}
\item For each subproblem $s$, the failure indicator correctly identifies whether or not $s$ is strongly feasible. 
\item If a subproblem $s$ is strongly feasible, then $s$ is in a strongly common-case configuration.
\end{itemize}

Before continuing, it is worth establishing three lemmas about the Common-Case Invariant:
\begin{lemma}
\label{lem:updatefriendly}
Suppose that the Common-Case Invariant holds. For a given level-$i$ subproblem $s$, we have with probability $1 - n_i^{-\omega(1)}$ that $s$ is strongly feasible, that $s$ is in a strongly common-case configuration, and that:
\begin{enumerate}[label=\normalfont(\roman*)]
\item\label{item:color} all of the elements in $s$'s buffer have colors $i$ and $i + 1$;
\item\label{item:fraction} if $s$ is neither the root nor a leaf, then the fraction of elements in $s$'s buffer that have color $i$ (as opposed to color $i + 1$) is in the range $[0.3, 0.7]$.
\end{enumerate}
Moreover, all of this is true even if we pick $s$ based in part on the hash of some specific element $x$ (this part will be straightforward, since knowing $h(x)$ can only change how many items hash to a given subtree by $\pm 1$).
\end{lemma}

We remark that, in Lemma \ref{lem:updatefriendly}, we are using $n_i^{-\omega(1)}$ to denote $n_i^{-f(n_i)}$ for some $f \in \omega(1)$, meaning that the $\omega$-notation is in terms of $n_i$ rather than, say, $n$. 

\begin{proof}
For a given level $i$ and level-$i$ subproblem $s$, define $t_s$ to be the sum of the sizes of the buffers of $s$ and all of $s$'s descendants; and define $q_s$ to be the total number of elements that hash to $s$ and $s$'s descendants. By construction, $\E[q_s] \in t_s - b_i \cdot [0.4, 0.6] \pm 1$ (the $\pm 1$ comes from the role of $x$ in the last sentence of the lemma statement). If $1 < i < \overline{\ell}$, then by a Chernoff bound, we know that with probability $1 - n_i^{-\omega(1)}$, we have $|q_s - \E[q_s]| \le \tilde{O}(\sqrt{n_i}) < 0.1 b_i$. Thus, if $i < \overline{\ell}$, then we have with probability $1 - n_i^{-\omega(1)}$ that 
\begin{equation}
q_x = \begin{cases}
    0 & \text{ if } i = 1 \\
    t_s - b_i \cdot [0.3, 0.7] & \text{ otherwise.}
\end{cases}
    \label{eq:qs}
\end{equation}

We can apply a union bound to conclude that, for any $1 \le i \le \overline{\ell}$ and any level-$i$ subproblem $s$, we have with probability $1 - f_i \cdot n_{i - 1}^{-\omega(1)} = 1 - n_i^{-\omega(1)}$ that \eqref{eq:qs} is true for all of $s$'s children (note that level-$1$ subproblems have $f_1 = 0$ children). Applying another union bound, we can conclude with probability
$$1 - \sum_{j \ge i} n_j^{-\omega(1)} = 1 - n_i^{-\omega(1)}$$
that \eqref{eq:qs} is true not just for all of $s$'s children but also for all of $s$'s parents' children, all of $s$'s grandparents' children, etc. We will assume this for the rest of the proof.

The fact that \eqref{eq:qs} holds for all of $s$'s children and $s$'s ancestors' children implies that all of $s$ and its ancestors are strongly feasible. By the Common Case Invariant, it follows that all of $s$ and its ancestors are in strongly common-case configurations.

The fact that both $s$ and its parent (if the parent exists) are in (even weakly) common-case configurations implies condition \ref{item:color}. Finally, if $s$ is neither the root nor a leaf, then condition \ref{item:fraction} follows from the second point along with \eqref{eq:qs}.
\end{proof}

\begin{lemma}
Let $T$ be the entire recursive tree and let $T'$ be a subtree. Suppose that every subproblem $s \in T \setminus T'$ that is strongly feasible is in a strongly common-case configuration, but that this is not necessarily the case for every $s \in T'$. Then, in $|T'| \log \log n$ time, one can rearrange the elements in $T'$ (in-place) so that the Common-Case Invariant holds; and so that, if there is a free slot in $T'$, it appears in the buffer of the root subproblem.
\label{lem:commoncasefix}
\end{lemma}
\begin{proof}
Let $s$ be the root subproblem of $T'$. We can check in time $O(|T'| + \log \log n)$ if $s$ is strongly feasible. If $s$ is not strongly feasible, then neither will any of its descendants be, and we can complete the lemma by (1) moving any free slot in $T'$ to be in $s$'s buffer and (2) setting all of the failure indicator bits of $s$ and its descendants to be true. If $s$ is strongly feasible, then $s$'s parent (if it has one) is in a strongly common-case configuration, which implies that all of the elements that hash to $s$ or $s$'s descendants are already in $T'$. Since $s$ is weakly feasible, it follows that we can rearrange the elements in $T'$ alone so that $s$ is in a weakly common-case configuration and so that any free slot in $T'$ appears in $s$'s buffer---this can be done in place in time $O(|T'|)$. At this point, because all of $s$ and its ancestors are in weakly common-case configurations, we can conclude that $s$ is, in fact, in a strongly common-case configuration (and we can update $s$'s failure indicator appropriately). Having placed $s$ in a strongly common-case configuration (and placed any free slot in $T'$ in $s$'s buffer), we can recurse on $s$'s children's subtrees in order to complete the lemma. It is straightforward to see that each of the $O(\log \log |T'|)$ layers of recursion takes time at most $O(|T'| + \log \log n)$ time, making for a total of $O(|T'| \log \log |T'| + (\log \log |T'|) \cdot \log \log n) = O(|T'| \log \log n)$ time.
\end{proof}

\begin{lemma}
Call a subproblem $s$ \defn{well supplied} if every (strict) ancestor $q$ of $s$ contains at least one element that hashes to $q$ in its buffer.
If the Common-Case Invariant holds, then the following modifications to the data structure cannot violate it:
\begin{enumerate}[label=\normalfont(\roman*)]
\item\label{item:insert} For a subproblem $s$ that is strongly feasible and contains a free slot in its buffer, insert an element $x$ satisfying $h(x) = s$ into that free slot.
\item\label{item:delete} For a subproblem $s$ that is strongly feasible and well-supplied, delete an element $x$ satisfying $h(x) \in \{s, \Parent(s)\}
$ from $s$'s buffer.
\item\label{item:move} Move an item $x$ between the buffers of the subproblem $h(x)$ that it hashes to and one of $h(x)$'s children.
\end{enumerate}
\label{lem:invariantmoves}
\end{lemma}
\begin{proof}
Critically, the types of modifications described in all three items have the properties that: they do not change for any subproblem in the tree whether the subproblem is in a weakly common-case configuration.

Thus, for subproblems that were already strongly feasible (and thus already in a strongly common-case configuration), those subproblems continue to be in strongly common-case configurations. We will complete the proof by showing that the modifications described in each bullet point do not change which subproblems in the tree are weakly feasible. This means that, if the Common-Case Invariant held before the modification, then it continues to hold after.

The insertion in \ref{item:insert} risks changing the weak feasibility of one of $s$'s ancestors. However, we know that each of $s$'s ancestors $q$ are already in strongly common-case configurations prior to the insertion, and we know that the insertion does not change this, so $s$'s ancestors are still in strongly (and thus weakly) common-case configurations. Since, after the insertion, there is no free slot in table, this implies that $s$'s ancestors are all still weakly feasible---their feasibility statuses have not changed after all.

The deletion in \ref{item:delete} also risks changing the weak feasibility of one of $s$'s ancestors. Once again, we know that $s$'s ancestors are already in strongly common-case configurations prior to the deletion, and we know that the deletion does not change this, so $s$'s ancestors are still in strongly (and thus weakly) common-case configurations. Furthermore, since $s$ is well-supplied, each of $s$'s ancestors $q$ contains at least one element in its buffer that hashes to $q$. So, via moves of the third type, we could move the free slot that is currently in $s$ (due to our deletion) to be in the root's buffer without changing which subproblems are in weakly common-case configurations. Thus it is possible for $s$ and its ancestors all to be in weakly common-case configurations while having the free slot appear only in the root's buffer. This, in turn, implies that $s$ and its ancestors are all still weakly feasible, so, once again, the feasibility statuses have not changed for any subproblems.

Finally, the modification in \ref{item:move} does not change the set of elements present overall. Thus it also cannot change for any subproblem whether that subproblem is weakly feasible.
\end{proof}

\paragraph{Implementing queries. }
We can now describe how to implement queries. Suppose we wish to query an element $x$, and let $s = h(x)$ be the level-$C(x)$ subproblem that $x$ hashes to.

To query $x$, we begin by querying $s$'s buffer and the buffers of $s$'s children. Since these buffers are each implemented as rainbow cells, this takes expected time
\[O\left(b_{C(x)}^{3/4} + f_{C(x)} \cdot b_{C(x) - 1}^{3/4}\right).\]
If any of these queries finds $x$, then we are done. Otherwise, we continue with the following logic.

Having examined $s$'s failure indicator, we know at this point whether or not $s$ is strongly feasible. If $s$ is strongly feasible, then by the Common-Case Invariant, $s$ must be in a strongly common-case configuration. This means that the only way $x$ could be in the hash table would be for it to appear in one of $s$'s or $s$'s children's buffers---since this is not the case, we can conclude that $x$ is not present. On the other hand, if $s$ is not strongly feasible, then we complete the query with the following \defn{failure-mode probe sequence}.

In the failure mode, the query scans the entirety of $s$ (including all of the buffers of all of $s$'s descendants). It then checks whether $s$'s parent $s^{(1)}$ is strongly feasible. If so, then the query is complete; otherwise, the query scans the entirety of $s^{(1)}$ (including all of the buffers of all of $s^{(1)}$'s descendants). It then checks whether $s^{(1)}$'s parent $s^{(2)}$ is strongly feasible. If so, then the query is complete; otherwise, the query scans the entirety of $s^{(2)}$, etc. The query continues like this until either it finds $x$ or it encounters an ancestor of $s$ that \emph{is} strongly feasible, at which point the query can conclude by the Common-Case Invariant that $x$ is not present.

Before continuing with our description of rainbow hashing, it is worth taking a moment to verify the correctness and running time of queries.

\begin{lemma}
Supposing the Common-Case Invariant, rainbow-hashing queries will be correct.
\end{lemma}
\begin{proof}
If $s = h(x)$ is strongly feasible, and thus is in a strongly common-case configuration, then every element that hashes to $s$ is guaranteed to be in the buffer of either $s$ or one of $s$'s children. Thus, by querying these buffers, the query will succeed.

Suppose $s \in h(x)$ is not strongly feasible. Let $y$ be the lowest ancestor of $s$ that is strongly feasible. (If $y$ does not exist, then the query scans the entire hash table and is thus necessarily correct.) Let $z$ be the child of $y$ whose subtree contains $s$. Then the query scans the entirety of $z$'s subtree, so it suffices to show that all $x$ satisfying $h(x) = s$ are contained in the subtree. The fact that $y$ is strongly feasible implies by the Common-Case Invariant that $y$ is in a strongly common-case configuration, which implies that every element $x$ that hashes to $z$ or $z$'s descendants is contained in the buffers of $z$ and its descendants. Therefore, if $x$ is in the hash table, it is contained in $z$ and its descendants, so the query will correctly ascertain whether $x$ is present.
\end{proof}

\begin{lemma}
Supposing the Common-Case Invariant, the expected time per query, overall, will be $O(1)$.
\end{lemma}
\begin{proof}
If we condition on $C(x)$, then since $s$ and $s$'s children are implemented as rainbow cells, the expected time to query $s = h(x)$'s and $s$'s children's buffers is
$$O\left(n_{C(x)}^{3/4} + f_{C(x)} \cdot n_{C(x) - 1}^{3/4}\right).$$
Recalling that, for $i > 1$, Lemma \ref{lem:pi} tells us that $\Pr[C(x) = i] = p_i = \Theta(b_{i - 1} / n_{i - 1})$, the expected time to query $s$'s and $s$'s children's buffers is
\begin{align*}
&    O\left(1 + \sum_{i > 1} \frac{b_{i - 1}}{n_{i - 1}} \cdot \left (b_{i}^{3/4} + f_{i} \cdot b_{i - 1}^{3/4}\right)\right) \\ 
& =  O\left(1 + \sum_{i > 1} \frac{n_i^{0.99 \cdot 0.51}}{n_i^{0.99}} \cdot \left ((n_i^{0.51})^{3/4} + n_i^{0.01} \cdot ((n_{i}^{0.99})^{0.51})^{3/4}\right)\right) \\
& =  O\left(1 + 1\sum_{i > 1} \frac{n_i^{0.99 \cdot 0.51}}{n_i} \cdot \left (n_i^{0.01} \cdot (n_i^{0.51})^{3/4}\right)\right) \\
& =  O\left(1 + \sum_{i > 1} n_i^{0.99 \cdot 0.51 - 1 + 0.01 + 0.51 \cdot 3/4}\right) \\
& =  O\left(1 + \sum_{i > 1} n_i^{-0.1}\right) \\
& = O(1).
\end{align*}

Additionally, to handle the case where $s$ may not be strongly feasible, the query will spend additional time $O(n_j)$, where $j$ is the largest $j$ such that the level-$j$ subproblem containing $s$ is not strongly feasible. Define $X_j$ to be the indicator random variable that the level-$j$ subproblem containing $s$ is not strongly feasible. Then our time contribution from cases where $s$ is not strongly feasible is at most
$$O\left(\sum_{j} X_j \cdot n_j\right),$$
which by Lemma \ref{lem:updatefriendly} has expectation
$$O\left(\sum_{j \ge i} n_j^{-\omega(1)} \cdot n_j\right) = O(1).$$

Thus the overall expected query time is $O(1)$.
\end{proof}

\paragraph{A helper method for updates: the ``sampling trick''. }Before we describe how to implement updates, let us first describe a sub-task that will prove useful. Given a level-$i$ subproblem $S$, and a color $\ell \in \{i, i + 1\}$, the $\Sample(s, c)$ protocol either returns an element in $s$'s buffer with color $\ell$ or declares that no such element exists. (If $i = \overline{\ell}$, then the only valid value for $\ell$ is $i$, and if $i = 1$ then the only valid value is $2$.)

The protocol is implemented by simply performing (up to) $b_i$ random samples from the buffer (returning if it ever finds an element with color $\ell$), and then, if none of those samples succeed, scanning the buffer in $O(b_i)$ additional time.

The following basic lemma will allow us to reason about the behavior of \Sample.
\begin{lemma}
Let $x$ be an element, let $s$ be the level-$i$ ancestor of $h(x)$ for some $i$, and let $\ell \in \{\max(i, 2), \, \min(i + 1, \overline{\ell})\}$. If the Common-Case Invariant holds, then $\Sample(s, \ell)$ takes $O(1)$ expected time.
\label{lem:sample}
\end{lemma}
\begin{proof}
By Lemma \ref{lem:updatefriendly}, we have with probability $1 - n_i^{-\omega(1)}$ that at least a constant-fraction of the elements in $s$'s buffer have color $\ell$. If this is the case, then the expected number of random samples needed to find such an element is $O(1)$. If this is not the case, which happens with probability $n_i^{-\omega(1)}$, then the \Sample procedure may take time as much as $O(b_i)$. Thus the overall expected time is
\[O(1) + O(n_i^{-\omega(1)} b_i) = O(1).
\qedhere\]
\end{proof}

\paragraph{Implementing insertions. } Suppose we wish to insert an element $x$. If the root $r$ is not strongly feasible, then we perform the insertion by rebuilding the entire table from scratch. Otherwise, we invoke a recursive function $\Insert(r, x)$ that we will now define. In general, the function $\Insert(s, x)$ takes two inputs:
\begin{itemize} 
\item a strongly feasible subproblem $s$ that currently contains the hash table's only free slot;
\item and an element $x$ that hashes to either $s$ or one of $s$'s descendants. 
\end{itemize}
It then implements the insertion of $x$ while preserving the Common-Case Invariant. The protocol for $\Insert(s, x)$ is:
\begin{enumerate}
    \item If $x$ hashes to $s$, then insert $x$ into $s$'s free slot, and return. By Lemma \ref{lem:invariantmoves}, this preserves the Common-Case Invariant.
    
    Otherwise, let $c$ be the child of $s$ on the path from $s$ to $h(x)$. Compute $y = \Sample(c, j)$, where $j$ is the level of $s$.
    \item If either $c$ is not strongly feasible or $y = \Null$, then rebuild $s$ from scratch to perform the insertion and preserve the Common-Case Invariant. This is possible by Lemma \ref{lem:commoncasefix}.
    \item Otherwise, move $y$ into the free slot in $s$, creating a free slot in $c$. (By Lemma \ref{lem:invariantmoves}, this preserves the Common-Case Invariant.) Now $c$ is a strongly feasible subproblem that contains the hash table's only free slot, so we can complete the insertion by calling $\Insert(c, x)$.
 \end{enumerate}

 \paragraph{Implementing deletions. } Now suppose we wish to delete an element $x$ (and create a free slot in the root subproblem). We will assume that we already know where the element is in the hash table, since this can be determined with an $O(1)$-expected-time query. If the root $r$ is not strongly feasible, then we perform the deletion by rebuilding the entire table from scratch. Otherwise, we invoke a recursive function $\Delete(r, x)$ that we will now define. In general, the function $\Delete(s, x)$ takes two inputs:
\begin{itemize} 
\item a strongly feasible subproblem $s$ that is well-supplied, as defined in Lemma \ref{lem:invariantmoves};
\item and an element $x$ that hashes to either $s$ or one of $s$'s descendants. 
\end{itemize}
It then implements the deletion of $x$ and creates a free slot in $s$'s buffer, all while preserving the Common-Case Invariant. The protocol for $\Delete(s, x)$ is:
\begin{enumerate}
    \item If $x$ is in $s$, then delete $x$ and return. By Lemma \ref{lem:invariantmoves}, this preserves the Common-Case Invariant.
    
    Otherwise, let $c$ be the child of $s$ on the path from $s$ to the subproblem whose buffer contains $x$. Compute $y = \Sample(s, j)$, where $j$ is the level of $s$.
    \item If either $c$ is not strongly feasible or $y = \Null$, then rebuild $s$ from scratch to delete $x$, place a free slot in $s$'s buffer, and preserve the Common-Case Invariant. This is possible by Lemma \ref{lem:commoncasefix}.
    \item Otherwise, since $s$ is well-supplied and $y$ exists, we can conclude that $c$ is well-supplied. Since, furthermore, $c$ is strongly feasible, we can legally invoke $\Delete(c, x)$. Doing so creates a free slot in some position $p$ of $c$'s buffer and (by induction) preserves the Common-Case Invariant. Finally, we move $y$ from $s$'s buffer to position $p$ of $c$'s buffer. This move creates a free slot in $s$'s buffer, as desired, and preserves the Common-Case Invariant by Lemma \ref{lem:invariantmoves}.
 \end{enumerate}

\paragraph{Analyzing insertions and deletions.}
Because the \Insert and \Delete protocols are so similar, we combine their analyses into a single lemma:
  \begin{lemma}
     The insertion and deletion protocols each take $O(\log \log n)$ expected time and preserve the Common-Case Invariant.
 \end{lemma}
 \begin{proof}
 The preservation of the Common-Case Invariant has already been established via Lemmas \ref{lem:commoncasefix} and \ref{lem:invariantmoves}.  So it suffices to prove the time bounds.

 The slow case for insertions/deletions is if the operation is forced to rebuild an entire subtree. This can happen if either the root $r$ is not strongly feasible, or if a call to either $\Insert(\cdot, \cdot)$ or $\Delete(\cdot, \cdot)$ triggers a rebuild of a subproblem in Step 2 of either protocol.

  Let $T_1$ denote the time spent on rebuilding entire subtrees and $T_2$ denote the other time spent on the insertion/deletion. The second quantity $T_2$ is dominated by making $O(\log \log n)$ calls to the $\Sample$ function, each of which we know by Lemma \ref{lem:sample} takes $O(1)$ expected time. So $\E[T_2] \le O(\log \log n)$. To complete the proof, we must also show that $\E[T_1] \le O(\log \log n)$.
 
 To bound $\E[T_1]$, we must first prove the following claim.
 \begin{claim}
  In the $\Insert(s, x)$ protocol, if Step 2 performs a rebuild because $\Sample(c, j) = \Null$, then $s$ is not strongly feasible after the insertion.

  Similarly, in the $\Delete(s, x)$ protocol, if Step 2 performs a rebuild because $\Sample(s, j) = \Null$, then $s$ is not strongly feasible after the deletion.
 \end{claim}
 \begin{proof}
  We begin by proving the claim for insertions. If $\Sample(c, j) = \Null$, then (prior to the insertion) there are no color-$j$ elements in $c$'s buffer. Since (prior to the insertion) $s$ is in a strongly common-case configuration, we know that all of the non-color-$j$ elements in the buffers of $c$ and its descendants hash to $c$ and its descendants. It follows that, after the insertion, the total number of items that hash to $c$ and its descendants will be $n_{j - 1} + 1$. It is therefore not possible for $s$ to be in a strongly (or even weakly) common-case configuration after the insertion, so $s$ is no longer strongly (or weakly) feasible.

  Next, we prove the claim for deletions. If $\Sample(s, j) = \Null$, then (prior to the deletion) there are no color-$j$ elements in $s$'s buffer (and therefore no elements that hash to $c$ or its descendants). Since (prior to the deletion) $s$ is in a strongly common-case configuration, all of the elements that hash to $s$ and $s$'s descendants are contained in the buffers of $s$ and $s$'s descendants. Since there are no such elements in $s$'s buffer, it follows that the total number of elements that hash to $s$ and $s$'s descendants is at most $n_i - b_i$. After the deletion, the number of such items will be at most $n_j - b_i - 1$, which means that the buffers of $s$'s descendants \emph{cannot} be occupied only by these items. It follows that (after the deletion) $s$ is no longer strongly (or weakly) feasible.
  \end{proof}

  From the preceding claim, we can conclude that, if either $\Insert(s, x)$ or $\Delete(s, x)$ performs a rebuild in Step 2, then at least one of $s$ or its child $c$ must be non-strongly-feasible either before or after the insertion or deletion. We know from Lemma \ref{lem:updatefriendly} that the probability of this happening for $s$ in a given level $j$ is at most $n_{j - 1}^{-\omega(1)} = n_j^{-\omega(1)}$. If a rebuild is performed, then it takes at most $\poly(n_j)$ time, so the expected contribution of each level $j$ to $T_1$ is at most 
  $$n_j^{-\omega(1)} \cdot \poly(n_j) \le O(1 / 2^j).$$
  This allows us to bound
  \[\E[T_1] \le \sum_j O(1/ 2^j) = O(1). \qedhere\]
\end{proof}

Putting the pieces together, we have proven Proposition \ref{prop:rainbow}.

\section{Dynamic Resizing Without Increasing Update Time}\label{sec:resizing}

In this section, we will extend Basic Rainbow Hashing to support dynamic resizing. As in the previous section, we shall continue to focus on hash tables that operate at load factor $1$. In this context, what dynamic resizing means is that, as the total number $ n $ of elements changes over time, the hash table automatically reconfigures itself to use exactly the first $n$ slots in memory. Perhaps surprisingly, we shall see that this seemingly stringent resizing property can be achieved without changing the timing characteristics of the hash table. In particular, we will prove the following proposition:

\begin{proposition}
The basic rainbow hash table can be extended to support dynamic resizing with a continual load factor of $1$, while preserving an expected query time of $O(1)$ and an expected update time of $O(\log \log n)$, where $n$ denotes the current size of the hash table. 
\label{prop:resizing}
\end{proposition}

The main challenge in resizing is to allow $n$ to change by a constant factor, that is to support $n$ changing within a range of the form, say, $[N, \, 1.1 \cdot N]$ for some $N$. This will be our focus for most of the section.

\paragraph{Embedding the recursion tree into an array with bit-reversed ordering.}
In order to describe our resizing approach, it will be helpful to adopt a specific layout for how we embed the buffers in the recursion tree into an array of size $ n $. We will refer to this layout as the \defn{bit-reversed layout} for reasons that will become clear shortly.

Without loss of generality, we can choose $b_i, m_i$ to all be powers of two for all $i > 1$. We can also defer any rounding errors to the leaf subproblems. That is, we can ensure that every level-$i$ subproblem, $i > 1$, has buffer size \emph{exactly} $b_i$ and fanout \emph{exactly} $f_i$, while allowing some leaf subproblems to have sizes that differ by $\pm 1$ from each other. Finally, just to simplify our discussion of the layout, we will think of every \emph{slot} in the bottom layer of the tree as representing its own subproblem (notice that this doesn't change the behavior of the data structure at all), and we will think of the fanout of any level-$2$ subproblem as being simply the sum of the sizes of the level-$1$ subproblems that it contains. (Note that, in doing this, we are implicitly redefining $f_2$ to be what was formerly $f_1 \cdot n_1$, we are redefining $n_1 = b_1 = 1$, and we are allowing the fanouts of level-$2$ subproblems to be within $1$ of $f_2$.)

With these WLOG assumptions in mind, place the buffers in levels $\overline{\ell}, \overline{\ell} - 1, \ldots$ in sub-arrays $A_{\overline{\ell}}, A_{\overline{\ell} - 1}, \ldots$ that appear one after another from left to right. Each $A_i$ with $i > 1$ will have a power-of-two size, but $A_1$ may not.

We label the subproblems in level $i$ using integers $0, 1, \ldots, m_i - 1$, and refer to them as ``the $j$-th subproblem in $A_i$'' for $0 \le j < m_i$.
For $i < \overline{\ell}$, it is tempting to declare the parent of the $j$-th subproblem in $A_i$ to be the $\lfloor j / f_{i + 1} \rfloor$-th subproblem in $A_{i + 1}$. Rather than using the standard layout, however, we will use a \defn{bit-reversed layout}: for $i < \overline{\ell}$, the parent of the $j$-th subproblem in $A_i$ is the $k$-th subproblem in $A_{i + 1}$, where $k$ is given by the \emph{low-order} $\log m_{i + 1}$ bits of $j$. Conversely, the $k$-th subproblem in $A_{i + 1}$ has as children the subproblems with indices of the form $r \cdot m_{i + 1} + k$ in level $i$. But important subtlety here is the relationship between levels $1$ and $2$. If level $1$ has size $m_1$, then the children of subproblem $k$ in level $2$ are the level-1 subproblems (i.e., slots in $A_1$) with indices
\[\{j = r \cdot m_2 + k \mid 0 \le j < m_1\}.\]
Conveniently, the restriction $j < m_1$ automatically handles rounding errors---it guarantees that the total number of subproblems/slots in level $1$ is exactly $m_1$, and dictates the assignment of those subproblems to level-2 parents. We can confirm that the layout handles rounding errors correctly, giving each level-$i$ subproblem the same total number of leaf slots up to $\pm 1$, with the following lemma.

\begin{lemma}
Using a bit-reversed layout, every level-$i$ subproblem has the same total number of leaf slots up to $\pm 1$.
\label{lem:bitreversed}
\end{lemma}
\begin{proof}
For the $k$-th level-$i$ subproblem in $A_i$, the leaf slots in $k$'s subtree are the slots in $A_1$ whose low-order bits are given by $k$, that is, the indices of the form 
\[\{j = r \cdot m_{i} + k \mid 0 \le j < m_1\}.\]
The number of such indices is exactly 
\[1 + \lfloor (m_1 - k) / m_i \rfloor.\]
Since $k \in \{0, 1, \ldots, m_i - 1\}$, this value varies by at most $\pm 1$ for different subproblems $k$ in $A_j$. 
\end{proof}

What is nice about this layout is not just that it handles rounding errors cleanly (a fact that should be viewed as a minor detail), but rather that, as we will now see, it enables a surprisingly simple resizing approach. 

\paragraph{Resizing by changing $m_1$ only.}
Suppose we wish to allow $n$ to change within the range $[N, \, 1.1 \cdot 
N]$ for some $N$. We will achieve this by simply changing the value of $m_1$ (i.e., the number of slots in $A_1$). Before an insertion, we first increment $n$ (and thus $m_1$). This creates a free slot in some leaf, which we can then migrate to the buffer of the root subproblem using the same protocol that we used to migrate free slots up the recursion tree for deletions. Similarly, after deletion, we decrement $n$ (and thus $m_1$). This removes a slot from some leaf---the element that gets evicted from that slot can be re-inserted using the standard insertion procedure (notice, in particular, that because we have just performed a deletion, there is a free slot in the root buffer of the tree). 

With these modifications in mind, the only point that we must be careful about is that, as $n$ changes within the range $[N, \, (1 + 0.1)N]$, the values of $p_i$ do \emph{not} change. For each subproblem $s$, let $t_s$ denote the size of the subproblem. Note that when we change $n$ and thus $m_1$, this also implicitly changes some values of $t_s$.

The only parts of the probabilistic analysis in Section \ref{sec:rainbowhashing} that use the relationship between the $p_i$'s and the other parameters are the proofs of \cref{lem:pi,lem:updatefriendly}. The fact that Lemma \ref{lem:pi} holds for $n = N$ directly implies that it holds for $n \in [N, \, 1.1 \cdot N]$. Lemma \ref{lem:updatefriendly} requires a bit more care, however, as it needs the $p_i$'s to satisfy \eqref{eq:corereq}, which for a level-$i$ subproblem $s$, with $1 < i < \overline{\ell}$, expands to
\begin{equation}(p_2 + \cdots + p_i)n / m_i \in t_s - [0.4, 0.6] \cdot b_i.
\label{eq:preq}
\end{equation}

To recover the proof of Lemma \ref{lem:updatefriendly}, it suffices to show that, as $n$ changes within the range $[N, \, 1.1 \cdot N]$, even though the values $\{p_i\}$ do not change, \eqref{eq:preq} continues to hold.
\begin{lemma}
Let $1 < i < \overline{\ell}$ and let $s$ be a level-$i$ subproblem. Let $T_s$ be the value of $t_s$ when $n = N$, and suppose that
$$(p_2 + \cdots + p_i)N / m_i = T_s - 0.5 \cdot b_i \pm 1.$$
Then, we claim that for all $n \in [N, \, 1.1 \cdot N]$, if we change $m_1$ so that the total number of slots is $n$ (and calculate the values of $t_s$ based on $n$), then we have
\[(p_2 + \cdots + p_i)n / m_i \in t_s - [0.4, 0.6] \cdot b_i.\]
\end{lemma}
\begin{proof}
Observe that
\begin{align*}
    & (p_2 + \cdots + p_i)n / m_i \\
    & = \frac{n}{N} (p_2 + \cdots + p_i) N / m_i \\
    & = \frac{n}{N} \cdot T_s - \frac{n}{N} \cdot 0.5 \cdot b_i \pm O(1) \\
    & \in \frac{n}{N} \cdot T_s - [0.5, 0.55] \cdot b_i \pm O(1).
\end{align*}
To complete the proof, we will show that
\[\frac{n}{N} \cdot T_s = t_s \pm O(1).\]
Since, by Lemma \ref{lem:bitreversed}, $t_s = n / m_i \pm 1$ and $T_s = N / m_i \pm 1$, it suffices to show that
$$\frac{n}{N} \cdot \frac{N}{m_i} = \frac{n}{m_i} \pm O(1),$$
which holds trivially.
\end{proof}

Having recovered Lemma \ref{lem:updatefriendly}, the rest of the analysis from Section \ref{sec:rainbowhashing} holds without modification. This gives us the following proposition:
\begin{proposition}
Given a parameter $N$, the basic rainbow hash table can be extended to support dynamic resizing with a continual load factor of $1$ and with $n$ varying in the range $[N, \, 1.1 \cdot N]$. Furthermore, this preserves an expected query time of $O(1)$ and an expected update time of $O(\log \log n)$, where $n$ denotes the current size of the hash table.
\label{prop:resizing0}
\end{proposition}

\paragraph{Allowing $n$ to change by more than a factor of $1.1$.}
Finally, we can use standard rebuilding techniques to allow for $n$ to change by more than a factor of $1.1$. Let $r \in (0.99, 1)$ be uniformly random. Let $N_i := \lfloor 1.09^i \cdot r \rfloor$. Whenever $n$ crosses from $\le N_i$ to $> N_i$ for some $i$, we will rebuild the entire hash table to use $N = N_i$. Such a rebuild can be performed in place and in $O(n \log \log n)$ expected time using Lemma \ref{lem:commoncasefix} (with $T'$ equal to the entire tree). Since each value $n$ has probability $\Theta(1/n)$ of being within $\pm 1$ of some $N_i$, the probability of a given update triggering a rebuild is $\Theta(1/n)$. The expected time spent per update on these rebuilds is therefore
$$\Theta(1/n) \cdot O(n \log \log n) = O(\log \log n).$$
Thus, we can extend Proposition \ref{prop:resizing0} to allow $n$ to change arbitrarily over time, while only adding $O(\log \log n)$ additional expected time per update, as desired. This completes the proof of Proposition \ref{prop:resizing}.

\section{Supporting Load Factor \texorpdfstring{$1 - \epsilon$}{1 - ε}}\label{sec:loadfactor}

In this section, we give a black-box transformation that takes the resizable rainbow hash table construction from Proposition \ref{prop:resizing} (which operates at load factor $1$) and uses it to construct a dynamically-resized hash table that operates at load factor $1 - \epsilon$, supports $O(1)$ expected-time queries, and supports $O(\log \log \epsilon^{-1})$ expected-time updates. 

Our main result will be the following theorem:
\begin{theorem}
There exists a classical open-addressed hash table that is dynamically resized to maintain a load factor of $\ge 1 - \epsilon$ while supporting queries in $O(1)$ expected time and updates in $O(\log \log \epsilon^{-1})$ time.
\label{thm:dynamic}
\end{theorem}
\noindent At the end of the section, we will also prove the analogous result for fixed-capacity hash tables.

Throughout the section, we will assume that $\epsilon = n^{-o(1)}$, since when $\epsilon = n^{-\Omega(1)}$, we are okay with $O(\log \log n)$-time updates, so we can keep the hash table at load factor $1$. Furthermore, as in Section \ref{sec:resizing}, it suffices to handle $n \in [N, \, 1.1 \cdot N]$ for some known parameter $N$, since we can use the random-threshold rebuild technique from Section \ref{sec:resizing} to handle larger changes in $n$. Finally, we will satisfy ourselves with a load factor of the form $1 - \Theta(\epsilon)$, since this is equivalent to $1 - \epsilon$ up to constant-factor changes in $\epsilon$.

\paragraph{The basic setup: rainbow hash tables with overflow handling.}
Our basic construction will be as follows. Let $k = \poly \epsilon^{-1}$.  We will maintain $N / k$ dynamically resized rainbow hash tables $H_1, H_2, \ldots, H_{N/k}$ that are each allocated space $(1 - \epsilon) \frac{n}{N} \cdot k \pm 1$ at any given moment. (The slots of the hash tables can be interleaved so that, if we wish to increase the space in each hash table by $1$, we just need to extend the size of our array overall by $N / k$.) As a convention, we will set $k' = \frac{n}{N} k$ (so $k'$ changes over time). 

Each element $x$ uses a random hash $g(x) \in [N/k]$ to select which hash table $H_{g(x)}$ it belongs in. If the hash table $H_{g(x)}$ is full when $x$ is inserted (which will be the common case), then $x$ is placed in a \defn{overflow buffer} $O_{g(x)}$ (whose implementation we will specify later). If an element in some $H_i$ is deleted, and the overflow buffer $O_i$ is non-empty, then an element in $O_i$ is moved to $H_i$. Similarly, if the amount of space allocated to $H_i$ is incremented, and $|O_{g(x)}| > 0$, then an item from $O_{g(x)}$ is moved to $H_i$ (so that the increase in the size of $H_i$ corresponds to an insertion from $H_i$'s perspective); and if the amount of space allocated to $H_i$ is decremented (and $H_i$ was full beforehand), then a random item will be deleted from $H_i$ (so that the decrease in size of $H_i$ corresponds to a deletion from $H_i$'s perspective) and that item will be placed in $O_i$.

The overall rule will be that, if $O_i$ is non-empty, then $H_i$ is full. Thus, if we use $z_i$ to denote the capacity of $H_i$ at any given moment, and $r_i$ to denote the number of elements $x$ satisfying $g(x) = i$ at any given moment, then $|O_i| = \max(0, \, r_i - z_i)$.

\begin{lemma}
With probability $1 - k^{-\omega(1)}$, we have
\[|O_i| \in [0.5 \cdot \epsilon k', \, 1.5 \cdot \epsilon k'].\]
Furthermore, $\E[\max(0, \, |O_i| - 2 \epsilon k)] = o(1)$.
\label{lem:overflow}
\end{lemma}
\begin{proof}
Let $o_i = X_i - |H_i|$, and note that $|O_i| = \max(0, o_i)$.  By design, $o_i = X_i - |H_i| = X_i - (1 - \epsilon)k' \pm O(1)$ where $X_i$ is a binomial random variable with mean $k'$. The amount that $o_i$ deviates from its mean of $\epsilon k' \pm 1$ is at most the amount that $X_i$ deviates from its mean of $k'$. By a Chernoff bound, the probability of $X_i$ deviating from its mean by $\Omega(\epsilon k') > \Omega((k')^{3/4})$ is at most $(k')^{-\omega(1)} = k^{-\omega(1)}$; it follows that, with probability $1 - k^{-\omega(1)}$, we have $o_i \in [0.5 \cdot \epsilon k', \, 1.5 \cdot \epsilon k']$ and therefore also $|O_i| \in  [0.5 \cdot \epsilon k', \, 1.5 \cdot \epsilon k']$.

To obtain the final claim in the lemma, we can again apply a Chernoff bound to $X_i$ to obtain that the expected value of $\max(0, \, X_i - \E[X_i] - \Omega(\epsilon k))$ is $o(1)$. Finally, since 
\[o_i - 2 \epsilon k = X_i - (1 - \epsilon)k' - 2 \epsilon k \pm 1 = X_i - \E[X_i] - 2 \epsilon k + \epsilon k' = X_i - \E[X_i] - \Omega(\epsilon k),\]
it follows that $\E[\max(0, \, |O_i| - 2\epsilon k)] = \E[\max(0, \, o_i - 2 \epsilon k)] = o(1)$.
\end{proof}

\paragraph{Handling under-filled $H_i$'s.}
If $H_i$ is full, at any given moment, it can be implemented directly as a resizable rainbow hash table. If $H_i$ is not full (which by Lemma \ref{lem:overflow} happens with probability $k^{-\omega(1)}$), then $H_i$ is said to incur an \defn{under-fill error}. In this case, slot $1$ of $H_i$ is left empty so that queries and updates can examine it and determine that an under-fill error has occurred. When an under-fill error has occurred, queries to elements $x$ satisfying $g(x) = i$ read all of $H_i$; and updates to $H_i$ spend $O(\poly |H_i|)$ time checking if $H_i$ is still experiencing an under-fill error, and rebuilding $H_i$ appropriately based on whether it is or is not still experiencing an under-fill error.

For any given element $x$, the probability that $H_{g(x)}$ is incurring an under-fill error at any given moment is $k^{-\omega(1)}$ by Lemma \ref{lem:overflow}. So the increase in expected query and update times due to under-fill errors is $o(1)$.

\paragraph{Implementing the overflow buffers.}
The only tricky part of the data structure is implementing the overflow buffers. We can afford to use $O(\epsilon N)$ total space to implement the $O_i$'s, and we wish to implement them in such a way that we can support the following operations in $O(1)$ expected time each:
\begin{itemize}
    \item $\Sample(O_i)$: samples a random element from $O_i$, if $|O_i| > 0$, and returns $\Null$ if $|O_i| = 0$.
    \item $\Query(O_i, y)$: determines if $y \in O_i$. 
    \item $\Insert(O_i, y)$: inserts $y$ into $O_i$.
    \item $\Delete(O_i, y)$: deletes $y$ from $O_i$.
\end{itemize}

Each $O_i$ is allocated a $2 \epsilon k$-size array $A_i$ that it uses to store its elements (unless $|O_i| > 2 \epsilon k$). The elements stored in $A_i$ treat $A_i$ as a linear-probing hash table.
By Lemma \ref{lem:overflow}, the load factor of $A_i$ is between $0.1$ and $0.9$ with probability $1 - k^{-\omega(1)}$. Thus, the expected time to perform queries/inserts/deletes in $A_i$ is $O(1)$. Additionally, if we wish to sample a random element from $A_i$, we can just randomly sample slots until we find one that is occupied; since, with probability $1 - k^{-\omega(1)}$ the number of occupied slots in $A_i$ is $\ge 0.1 \cdot |A_i|$, the expected time of this sampling procedure is $O(1)$.

The only question is what we should do when $A_i$ itself overflows, that is, when $|O_i| \ge |A_i| = 2 \epsilon k$. Note that, since this occurs with only a small probability ($k^{-\omega(1)}$), we are okay with having relatively expensive (say, $\poly(k)$-time) operations in this case.

To handle overflowed $A_i$'s, we allocate an array $B$ of size $\epsilon N$. If $A_i$ overflows, its overflow elements $x$ treat $B$ as a linear-probing hash table, where the hash of $x$ within $B$ is calculated by a ``hash function''
\[\overline{h}(i) := \epsilon \cdot k \cdot i.\]
Note that, since $A_i$ is indexed by $i$ ranging over $i \in [N/k]$, the quantity $\overline{h}(i)$ is a valid index in $[|B|] = [\epsilon N]$. 

Since all of the overflow elements $x$ from $A_i$ have the same hash $\overline{h}$ as each other, it suffices to show that, even if we condition on $A_i$ overflowing, the expected length of the run in $B$ that contains the overflow elements is at most $\poly(k)$. This will allow us to implement operations on $O_i$ in $\poly(k)$ expected time when $A_i$ overflows, as desired.

Thus, to complete our discussion of how to implement the $O_i$'s, it suffices to prove the following lemma:
\begin{lemma}
    Conditioned on $A_i$ overflowing, the expected length of the run in $B$ that contains the overflow elements from $A_i$ is at most $\poly(k)$.
\end{lemma}
\begin{proof}

It suffices to show that, conditioned on $A_i$ overflowing, we have for each run-length $t > \poly(k)$ (where $\poly(k)$ is a large polynomial of our choice) that: the probability that the overflow elements from $A_i$ are in a run in $B$ of length $t$ is at most
\[
e^{-\omega(\log t)}. \numberthis \label{eq:prob_run_len_t}
\]
In order for the elements to be in a run of length $t$, there must be some contiguous interval $I \ni \overline{h}(i)$ of $t$ slot indices in $B$ such that the number of elements in $B$ that hash into $I$ is exactly $|I| = t$. As there are only $t$ options for $I$ satisfying $|I| = t$ and $\overline{h}(i) \in I$, it suffices to show that each individual option has probability at most $e^{-\omega(\log t)}$ of occurring. Since $t > \poly(k)$ for a polynomial of our choice, this, in turn, reduces to bounding the probability of a given interval $I$ occurring by, say,
\[e^{-\Omega(\sqrt{t} / \poly(k))}.\]
For the rest of the proof, fixing some interval $I$ of $t$ slots in $B$ containing $\overline{h}(i)$, we wish to bound the probability that, conditioned on $A_i$ overflowing, there are at least $t$ elements that overflow from $A_j$'s satisfying $\overline{h}(j) \in I$.

Define $X_j$ as the number of items that hash to $H_j$. (So, if we were to not condition on anything, then $X_j$ would be a binomial random variable with mean $k'$.) Define $Y_j$ as the number of items that overflow from $A_j$, that is,
\[Y_j = \max(0, \, X_j - |A_j| - |H_j|).\]
Recall that we are conditioning on $Y_i \ge 1$.

There are $O(\epsilon t / k)$ values of $j$ such that $\overline{h}(j) \in I$. Let $J$ be the set of such values $j$, excluding $i$. Then, in order for $I$ to have $t$ elements in it, we would need
\[Y_i + \sum_{j \in J} Y_j \ge t,\]
which implies either that
\begin{equation} 
Y_i \ge t/2
\label{eq:Y1}
\end{equation}
or that
\begin{equation}
\sum_{j \in J} Y_j \ge t/2.
\label{eq:Y2}
\end{equation}
Again, we are interested in the probabilities of these events conditioned on $Y_i \ge 1$.

To bound the probability of \eqref{eq:Y1}, observe that the random variable $(Y_i - 1 \mid Y_i \ge 1)$ is dominated by $X_i$ (not conditioned on anything).\footnote{This is an application of the more general fact that, for a binomial random variable $X$, $\Pr[X \ge r \mid X \ge r - 1]$ is a monotonically decreasing function in $r$.} Therefore,
\[\Pr[Y_i \ge t/2 \mid Y_i \ge 1] \le \Pr[X_i \ge t/2 - 1].\]
Since $X_i$ is a binomial random variable with mean $\Theta(k)$, this latter probability is, by a Chernoff bound, at most $e^{-\Omega(t / k)}$. 

To bound the probability of \eqref{eq:Y2}, observe that 
\[\Pr\left[\sum_{j \in J} Y_j \ge t/2 \;\middle|\; Y_i \ge 1\right] < \Pr\left[\sum_{j \in J} Y_j \ge t/2\right],\]
where the latter probability does not condition on anything.

The values $\{Y_j \mid j \in J\}$ are negatively associated random variables that each has expected value $o(1)$ (by Lemma \ref{lem:overflow}) and that are each bounded above by a geometric random variable with mean $O(k)$ (since even $X_j$ is, by even a very weak Chernoff bound, bounded above by such a geometric random variable). By a Chernoff bound for sums of negatively associated geometric random variables \cite{wajc2017negative}, the probability that $\sum_{j \in J} Y_j$ exceeds its mean by more than $r \cdot \sqrt{|J|} k$, for a given $r > 0$, is at most $e^{-\Omega(r)}$. Since the mean of this sum is $|J| \cdot o(1) \le |J|$, it follows that 
\[\Pr\left[\sum_{j \in J} Y_j > |J| + r \cdot \sqrt{|J|} k\right] \le e^{-\Omega(r)}.\]
Since $|J| = O(\epsilon t / k) = o(t)$, this lets us bound 
\begin{align*}
    \Pr\left[\sum_{j \in J} Y_j \ge t/2\right] & 
    < \Pr\left[\sum_{j \in J} Y_j \ge |J|  + t/4\right] 
     \\ 
    & = e^{-\Omega(t / (k\sqrt{|J|}))} \\
    & = e^{-\Omega(t / (k\sqrt{\epsilon t / k}))} \\
    & = e^{-\Omega(\sqrt{t} / \poly(k))},
\end{align*}
as desired. 
\end{proof}

In the case where $B$ overflows, which is an extremely rare event, we say the \defn{global failure} has occurred. When this happens, we keep the entire hash table in an arbitrary state (i.e., placing all keys in arbitrary slots), and for each insertion/deletion/query, we spend $\Theta(n)$ time scanning through the table and checking if it can return to the normal case.
A special boolean is stored to indicate whether the global failure is occurring, which is again encoded by the relative order of two special elements. According to \eqref{eq:prob_run_len_t}, at any given moment, the probability that $B$ overflows is at most $e^{-\omega(\log (\eps N))} = N^{-\omega(1)}$, so the global failure contributes a negligible amount to the expected insertion/deletion/query time.

\paragraph{Putting the pieces together.} 
We can now summarize the full procedure for implementing insertions/deletions/queries.

To query an element $x$, we must query $O_{g(x)}$ and $H_{g(x)}$. The expected time needed to query $H_{g(x)}$ is $O(\log \log k) = O(\log \log \epsilon^{-1})$; and the expected time needed to query $O_{g(x)}$ is $O(1)$. To keep the query within the model of classical open addressing, we can simply interleave the two probe sequences and stop when both queries have completed. Overall, the expected query time is $O(1)$.

To insert an element $x$, we first check if $H_{g(x)}$ is under-filled. If so, we follow the protocol described earlier to handle under-fill errors. Otherwise, we insert $x$ into $O_{g(x)}$. Each of the above steps takes $O(\log \log k) = O(\log \log \epsilon^{-1})$ expected time.

To delete an element $x$, we first determine where it is. If $x \in O_{g(x)}$, we delete it in $O(1)$ expected time. Otherwise, if $x \in H_{g(x)}$, we delete it from $H_{g(x)}$; we use $\Sample(O_{g(x)})$ to find an element $y$ that we can move from $O_{g(x)}$ into $H_{g(x)}$; and we insert $y$ into $H_{g(x)}$ in order to keep $H_{g(x)}$ at load factor $1$. If $y$ does not exist, then the free slot in $H_{g(x)}$ is placed in its slot $1$ to indicate that it is experiencing an under-fill error. Each of the above steps takes $O(\log \log k) = O(\log \log \epsilon^{-1})$ expected time.

Finally, as insertions and deletions are performed, we must not only implement those operations, but also add/remove slots to the $H_i$'s in order to maintain the invariant that each $H_i$ is given $(1 - \epsilon) \frac{n}{N} \cdot k \pm 1$ slots at any given moment.

Adding a slot to some $H_i$ is implemented as follows. A call to $\Sample(O_i)$ is made to either get a random element of $O_i$ or determine that $|O_i| = 0$. If $|O_i| = 0$, then the addition of a slot to $H_i$ will cause it to be under-filled. In this case, $H_i$ is rebuilt from scratch (but remember that, since under-filled $H_i$'s are very rare, this contributes negligibly to our expected time bound). In the more common case where $|O_i| > 0$, we move an item from $O_i$ into $H_i$---this allows $H_i$ to view its increase in size as being due to an insertion, which in turn takes $O(\log \log k) = O(\log \log \eps^{-1})$ expected time.

Deleting a slot from some $H_i$, on the other hand, is implemented as follows. If $H_i$ is under-filled, then $H_i$ is simply rebuilt (again, since under-filled $H_i$s are very rare, this contributes negligibly to our expected time bound). Otherwise, a random element is removed from $H_i$ and placed into $O_i$. This allows $H_i$ to view its decrease in size as being caused by a deletion, which in turn takes $O(\log \log k) = O(\log \log \eps^{-1})$ expected time.

Putting the time bounds together, we have proven Theorem \ref{thm:dynamic}.

\paragraph{Supporting all load factors $< 1 - \epsilon$ for the fixed-capacity case.}
Finally, we conclude the section by describing how to handle arbitrary load factors in a fixed-capacity hash table. 

\begin{theorem}
There exists a classical open-addressed hash table that has a fixed capacity $N$, and that allows for load factors of up to $1 - \epsilon$ while supporting queries in $O(1)$ expected time and updates in $O(\log \log \epsilon^{-1})$ time.
\label{thm:fixed}
\end{theorem}
\begin{proof}
Let us begin by assuming that the size is guaranteed to stay in $[0.95 N, \, (1 - \epsilon) N]$. Then, we will implement our fixed-capacity data structure, which we will call $D$, by actually using the dynamically resized data structure, which we will call $D'$, already described earlier in the section. The data structure $D'$ lives in a prefix of the memory allocated to $D$, and is parameterized to preserve a load factor of $1 - \epsilon$ as the number of elements varies within the range $[N', \, (1 - \epsilon)N < 1.1 N']$ for $N' = 0.95 N$. As we have already established, $D'$ supports $O(1)$ expected-time queries and $O(\log \log \epsilon^{-1})$ expected-time updates.

There is one point that we must be very careful about, however. One cannot, \emph{in general}, use a dynamically-resized classical-open-addressed hash table to implement a fixed-capacity classical open-addressed hash table. This is because, in the dynamic-resizing case, the probe sequence for an item $x$ is permitted to depend on both $x$ and the \emph{current} value of $N$; but in the fixed-capacity case, the probe sequence must depend on only $x$ and a \emph{fixed} $N$. Fortunately, in our construction of $D'$, the probe sequences that we use for $n \in [N', \, 1.1\cdot N']$ are subsequences of the probe sequences that we use for $n = 1.1 \cdot N$ (and the difference for each probe sequence is just the addition/removal of $O(1)$ probes in the first $O(1)$ entries of the sequence). Therefore, in our case, we \emph{can} use $D'$ within $D$.

Finally, we must also handle cases where our size drops below $0.95 N$. Let $T$ be a random threshold in $[0.9 N, \, 0.95 N]$. Whenever the number of elements in the hash table crosses below $T$, we rebuild the hash table to use standard linear probing. Whenever the number of elements crosses above $T$, we rebuild the hash table to use $D'$ as described above. Each rebuild takes $O(N \log \log \epsilon^{-1})$ time, and the probability of a given insert/deletion crossing the threshold $T$ is $O(1/ N)$, so the contribution of rebuilds to the expected update time is $O(\log \log \epsilon^{-1})$.
\end{proof}

\section{The Lower Bound}\label{sec:lower}

In this section, we prove an $\Omega(\log \log \epsilon^{-1})$ lower bound on the amortized expected time per insertion/deletion in any classical open-addressing hash table that supports (even moderately) efficient queries.

\begin{restatable}{theorem}{thmlower}
  \label{thm:lowerbound/nonuniform}
  Suppose the universe size $U = \poly n$ is a large polynomial of $n$.
  If a classical open-addressing hash table stores $n$ keys with load factor $1 - \eps$, then the expected amortized time per operation is at least $\Omega(\log \log \eps^{-1})$. Moreover, as long as the expected query time is $O\bk[\big]{2^{\sqrt{\log \eps^{-1}}}}$, the expected amortized time per insertion/deletion is at least $\Omega(\log \log \eps^{-1})$.
\end{restatable}

The hard distribution used for proving Theorem \ref{thm:lowerbound/nonuniform} will simply be a sequence of $n^2$ random insertions and deletions  (\cref{dist:hard}). Under this operation sequence, we will show that, if a classical open-addressing hash table has low average probe complexity, it must relocate a large number of keys to other slots during the operations (hence giving it a large insertion/deletion time). By Yao's minimax principle, we also assume the hash table is deterministic.

\begin{center}
  \begin{algorithm}[H]
    \AlgorithmCaption{Distribution}{Hard distribution}{dist:hard}
    \DontPrintSemicolon

    Initialize the hash table with $n$ random keys from the universe $[U]$\;
    \Repeat{$M = n^{2}$ \textup{times}}{
      Delete a random element from the current key set $S$\;
      Insert a random element in $[U] \setminus S$\;
    }
  \end{algorithm}
\end{center}

Recall that $n$ is the maximum number of keys the hash table can store, and assume the load factor is $1 - \eps$ for some $1/n \le \eps \le \Theta(1)$, so the number of slots is $N = (1 + \Theta(\eps)) \cdot n$. There is a \emph{deterministic} function $h$ that maps each key $x \in [U]$ to a probe sequence $\BK{h_i(x)}_{i \ge 1}$. For technical reasons, we allow each entry $h_i(x)$ to be either a slot $s \in [N]$ or ``$\Null$''. Without loss of generality, we assume that there is a \emph{special slot}, say slot $N$: When we insert any key $x$, key $x$ is first put into the special slot, then the algorithm will arrange the keys to move $x$ to a normal slot.

Recall that a key $x$ storing in slot $s = h_i(x)$ is said to have \defn{probe complexity} $i$ (assuming $h_i(x)$ is the first occurrence of $s$ in the probe sequence). When a key is stored in the special slot, we say the key has probe complexity $N$. Any key with probe complexity $i$ will cost the query algorithm $O(i)$ time. Thus, the \emph{average probe complexity} over the current key set measures the query time of the hash table. During insertions and deletions, the hash table may move some keys to other slots, and we define the \defn{switching cost} of this operation to be the number of moved keys. The switching cost is a lower bound of the time spent on the operation.

\smallskip

For any probe-sequence function $h$, integer $i \ge 1$, and slot $s \in [N]$, we define
\[
  q(h, i, s) \defeq n \Pr_{x \in [U]} [\ProbeComplexity(x, s) \le i] = n \Pr_{x \in [U]} [h_k(x) = s \ \textup{for some}\ k \le i].
\]
We say $h$ is \defn{nearly uniform} if $q(h, i, s) \le O(i^{10})$ for all $i$ and $s$. We first assume the function $h$ is nearly uniform and prove the following lower bound. At the end of this section, we will remove this assumption.

\begin{theorem}
  \label{thm:lowerbound/uniform}
  Suppose the universe size $U = \poly n$ is a large polynomial of $n$. Assume there is a classical open-addressing hash table, which stores $n$ keys with load factor $1 - \eps$, uses a nearly uniform probe-sequence function $h$, and has average probe complexity $O\bk[\big]{2^{\sqrt{\log \eps^{-1}}}}$ in expectation at any given moment, then the expected amortized switching cost during each insertion or deletion must be $\Omega(\log \log \eps^{-1})$.
\end{theorem}

To prove this theorem, we let the hash table take \cref{dist:hard} as input, and show a lower bound on the average switching cost per insertion and deletion. We start by setting up main concepts in our analysis.

\paragraph{Levels.}
We define $L \defeq \ceil{(\log \log \eps^{-1}) / 2}$. Suppose a key $x$ is stored in slot $s = h_k(x)$ with probe complexity $k$, i.e., $k = \min \BK{k' \in \N_+ \mid h_{k'}(x) = s}$. We define the \defn{level} of key $x$ stored in slot $s$, written $\l(x, s)$, as follows:
\begin{itemize}
\item If $k \le 2^{2^{L}}$, the level is $\l(x, s) = 0$.
\item If $k > 2^{2^{2L}}$, the level is $L$. As a special case, any key stored in the special slot has level $L$.
\item Otherwise, if $k \in (2^{2^{L + i - 1}}, 2^{2^{L + i}}]$ for some $i \in \N_+$, the level is $i$.
\end{itemize}
Moreover, we define the level of a slot $s$ in a given state to be the level of the key stored in slot $s$, if slot $s$ is not empty; the level of an empty slot is defined as $L$. It is clear that the level of the special slot is always $L$.

When the algorithm moves a key $x$ from one slot to another, it may change the level of the key. We define the \defn{impact} of the move to be how much the level of $x$ \emph{decreases} (negative impact means the level increases). We have the following lemma from \cite{benderhashing} with almost the same proof.

\begin{lemma}[{\cite[Lemma 5]{benderhashing}}]
  \label{lem:impact}
  Let $\Psi$ be the sum of impacts of all the moves that the algorithm performs during the $M$ insertions and deletions. Then
  \[
    \E[\Psi] = \Theta(ML).
  \]
\end{lemma}

\begin{proofsketch}
  Let $J$ be the sum of levels of all keys stored in the current hash table. When we insert a key $x$, it is first placed in the special slot with level $L$, thus increasing $J$ by $L$. When we delete a key $x$, the expected level of $x$ is $O(1)$ since the expected average probe complexity at any moment is small, so removing $x$ from the slot containing it will decrease $J$ by $O(1)$ in expectation. Therefore, the $M$ insertions and deletions will increase $J$ by $\Theta(ML)$, implying that the hash table should rearrange the keys to decrease the sum of levels by $\Theta(ML)$. The lemma follows.
\end{proofsketch}

\paragraph{Potential function.}
To prove \cref{thm:lowerbound/uniform}, we will construct a potential function $\Phi$ of any given state of the hash table, satisfying three properties:
\begin{itemize}
\item \textbf{Property 1.} Each insertion or deletion increases $\Phi$ by $O(1)$ in expectation.
\item \textbf{Property 2.} If the algorithm performs a move with impact $r$ (i.e., the level of the moved key is decreased by $r$), $\Phi$ will be decreased by $r \pm O(1)$.
\item \textbf{Property 3.} $0 \le \Phi \le O(nL)$ always holds.
\end{itemize}

Once we have a potential function $\Phi$ satisfying all three properties, the following proposition from \cite{benderhashing} will imply \cref{thm:lowerbound/uniform}.

\begin{proposition}[Implicit in \cite{benderhashing}]
  \label{prop:prop-imply-thm}
  If \cref{lem:impact} holds and there is a potential function $\Phi$ satisfying Properties 1 to 3 above, then \cref{thm:lowerbound/uniform} holds.
\end{proposition}

\begin{proofsketch}
  At a given moment, let $\Psi$ be the total impact the algorithm has made so far, and let $\Phi$ be the potential function on the current state. By \cref{lem:impact} and Property 3 of $\Phi$, the value $\Psi + \Phi$ is increased by $\Theta(ML)$ during all $M$ insertions and deletions, but each operation and each move can only increase it by $O(1)$. Therefore, the number of moves during all operations is at least $\Omega(ML)$.
\end{proofsketch}

The only remaining step to prove \cref{thm:lowerbound/uniform} is to construct such a potential function.

\subsection{Constructing the potential function}

In this subsection, we construct the potential function $\Phi$ based on the concept of \emph{stanzas}.

\begin{definition}[Stanzas]
  Fix a state of the hash table. Let $10 \le i \le L$ and $j \ge 2$ be integers. We say a sequence of slots $s_1, s_2, \ldots, s_j$ is an $i$-\defn{stanza}, if the following conditions hold:
  \begin{itemize}
  \item $s_1$ and $s_j$ have level $\ge i$, while $s_2, \ldots, s_{j-1}$ are \emph{distinct} slots with level $\le i - 10$ (we allow $s_1 = s_j$).
  \item Slots $s_1, \ldots, s_{j - 1}$ are non-empty. Assuming keys $x_1, \ldots, x_{j - 1}$ are stored in slots $s_1, \ldots, s_{j - 1}$ respectively, then $\l(x_k, s_{k + 1}) \le i - 10$ for every $k \in [j - 1]$.
  \end{itemize}
  In such a stanza, we call $s_1$ the \defn{starting slot}, $s_2, \ldots, s_{j - 1}$ the \defn{internal slots}, and $s_j$ the \defn{final slot}. We say a collection of stanzas is \defn{disjoint} if each slot with level $\ge i$ is used at most once as the starting slot and at most once as the final slot, whereas each slot with level $\le i - 10$ is used at most once as the internal slot.
\end{definition}

\begin{definition}[Potential of stanzas]
  \label{def:potential-stanza}
  For an $(i + 10)$-stanza $s_1, \ldots, s_j$, assuming $x_k$ is the key stored in slot $s_k$ for $1 \le k \le j - 1$, we define its \defn{potential} to be
  \[
    \phi(s_1, \ldots, s_j) \defeq 1 - \sum_{k = 2}^{j - 1} \sqrt{2}^{\l(x_k, s_k) - i} - \sum_{k = 1}^{j - 1} \sqrt{2}^{\l(x_k, s_{k+1}) - i}.
  \]
\end{definition}

\begin{definition}[Potential function]
  We define the potential of a collection of disjoint stanzas to be the sum of potentials of the stanzas in the collection. For $10 \le i \le L$, let $\Phi_i$ denote the maximum potential of a collection of disjoint $i$-stanzas in the current state. We define the potential function
  \[
    \Phi \defeq \sum_{i = 10}^{L} \Phi_i.
  \]
\end{definition}

To analyze the properties of $\Phi$, we define the \defn{key-slot graph} $G$ as follows.

\begin{definition}[Key-slot graph]
  At any given moment, we use $A_1$ to denote the set of $n$ keys in the current key set; use $A_2$ to denote $n$ random keys in $[U] \setminus A_1$, in which we will insert a random one in the next insertion. The \defn{key-slot graph} of the given state is a bipartite graph $G$ with $A \defeq A_1 \cup A_2$ being the left-vertices and the set $[N]$ of all slots being the right-vertices. For each key $a \in A$ and each slot $s \in [N]$, if $\l(a, s) < L$, there is an edge in $G$ connecting them with level $\l(a, s)$.
\end{definition}

\begin{definition}
  For each vertex $x \in G$, we define its \defn{level-$i$ degree} as the number of edges with level \emph{at most} $i$ associated with $x$, written $\deg_i(x)$. We say vertex $x$ is \defn{low-degree}, if for all $0 \le i < L$, there is $\deg_i(x) \le 2^{2^{L + i + 5}}$. Otherwise, we say $x$ is \defn{high-degree}.
\end{definition}

Note that the level-$i$ degree of a \emph{left-vertex} $a \in A$ is at most $2^{2^{L + i}}$ deterministically, so all left-vertices are low-degree. For right-vertices, the expected level-$i$ degree is also bounded by $\Theta\bk[\big]{(2^{2^{L + i}})^{10}}$ since the probe-sequence function is nearly uniform.

\begin{lemma}
  \label{lem:right-degree}
  Let $S$ be the set of high-degree (right-)vertices in $G$. With probability $1 - 1 / \poly L$, the total number of neighbors of vertices in $S$ is
  \[
    \sum_{s \in S} \deg_{L - 1}(s) \le \frac{n}{2^{\Omega(2^{2^{L}})}}.
  \]
\end{lemma}

\begin{proof}
  We fix a right-vertex $s \in [N]$ and analyze its level-$i$ degree. Since $U$ is a large polynomial of $n$, without loss of generality, the left-vertices can be viewed as $2n$ independently random keys from $[U]$, written $\{a_1, \ldots, a_{2n}\}$. Let $X_k \defeq \ind[\l(a_k, s) \le i]$, then
  \[
    \E[X_k] = \Pr_{a_k \in [U]}\Bk[\big]{\textup{probe-complexity}(a_k, s) \le 2^{2^{L + i}}} \le \Theta\bk[\bigg]{\frac{(2^{2^{L + i}})^{10}}{n}} \le \frac{2^{2^{L + i + 4}}}{n}
  \]
  by the near-uniformity of the probe-sequence function.
  Then, $X \defeq \sum_{k=1}^{2n} X_k$ is the level-$i$ degree of $s$, whose expectation does not exceed $O\bk{2^{2^{L + i + 4}}}$. By a Chernoff bound, for each $D \ge \Omega(2^{2^{L + i + 4}})$,
  \[
    \Pr[X > D] \le 2^{-\Omega(D)}. \numberthis \label{eq:chernoff}
  \]
  Substituting $D = \Theta(2^{2^{L + i + 5}})$ into the inequality, we know
  \[
    \Pr[X > 2^{2^{L + i + 5}}] \le 2^{-\Omega(2^{2^{L + i + 5}})} \le 2^{-\Omega(2^{2^L})}.
  \]
  Taking a union bound over $0 \le i < L$, we know that
  \[
    \Pr[\textup{$s$ is high-degree}] \le 2^{-\Omega(2^{2^L})} \cdot L = 2^{-\Omega(2^{2^L})}. \numberthis \label{eq:high-degree-prob}
  \]
  Moreover, taking summation of \eqref{eq:chernoff} over all integers $D > 2^{2^{L + i + 5}}$, we know
  \[
    \E\Bk[\Big]{\deg_i(s) \cdot \ind\Bk[\big]{\deg_i(s) > 2^{2^{L + i + 5}}}} \le 2^{-\Omega(2^{2^L})}. \numberthis \label{eq:sum-high-degree}
  \]

  To bound the number of neighbors of $S$ (the set of high-degree nodes), we divide $S$ into two parts: Let $S_1$ be the set of vertices $x$ satisfying $\deg_{L - 1}(x) > 2^{2^{2L + 4}}$, i.e., they are high-degree because their level-$(L - 1)$ degrees are too large; let $S_2 = S \setminus S_1$. We bound the number of neighbors of the two parts separately.

  For $S_1$, we take summation of \eqref{eq:sum-high-degree} over all right-vertices with $i = L - 1$, obtaining
  \[
    \E\Bk*{\sum_{s \in S_1} \deg_{L - 1}(s)} \le n \cdot 2^{-\Omega(2^{2^L})}. \numberthis \label{eq:deg-s1}
  \]

  For $S_2$, we take summation of \eqref{eq:high-degree-prob} to obtain
  \[
    \E\Bk*{\sum_{s \in S_2} \deg_{L - 1}(s)} \le \E\Bk*{|S_2|} \cdot 2^{2^{2L + 4}} \le n \cdot 2^{-\Omega(2^{2^L})} \cdot 2^{2^{2L + 4}} = n \cdot 2^{-\Omega(2^{2^L})}. \numberthis \label{eq:deg-s2}
  \]

  Adding \eqref{eq:deg-s1} and \eqref{eq:deg-s2} together and applying a Markov inequality, the lemma follows.
\end{proof}

\paragraph{\boldmath $i$-short paths.}

Assume $s_1, \ldots, s_{j}$ is an $(i + 10)$-stanza with positive potential, and $x_k$ is stored in slot $s_k$ for $k \in [j - 1]$. Then, there must be a path in $G$ from $x_1$ to $s_j$ consisting of:
\begin{itemize}
\item no edges with level $> i$;
\item at most 1 edge with level $i$;
\item at most $\sqrt{2}$ edges with level $i - 1$;
\item at most $\sqrt{2}^2$ edges with level $i - 2$;
\item $\cdots$
\end{itemize}
Formally, for every integer $0 \le \delta \le i$, the number of edges with level $i - \delta$ in the path cannot exceed $\sqrt{2}^{\delta}$. We say a path is \defn{$i$-short} if it satisfies this condition.

Next, we show that a random key $a_1 \in A_1$ or $a_2 \in A_2$ is unlikely to reach any right-vertex (slot) $s \in [N]$ with level $\ge i + 10$ via an $i$-short path.

\begin{lemma}
  \label{lem:path}
  Let $a_1 \in A_1$ and $a_2 \in A_2$ be random keys, and $0 \le i \le L - 10$ be an integer. With probability $1 - 1 / \poly L$, neither $a_1$ nor $a_2$ can reach a right-vertex (slot) in the key-slot graph $G$ with level $\ge i + 10$ via an $i$-short path.
\end{lemma}

\begin{proof}
  We first build an auxiliary graph $G'$ from $G$ as follows:
  \begin{itemize}
  \item for every high-degree vertex $s$, we color all its neighbors in black;
  \item for each right-vertex $s \in [N]$ with level $\ge i + 10$, we color it in black;
  \item we delete all high-degree vertices.
  \end{itemize}
  After these modifications, we call the obtained graph $G'$, which contains a subset of vertices colored black. There are two properties of $G'$:
  \begin{enumerate}
  \item every node in $G'$ is low-degree;
  \item if a node $a \in A$ can reach a slot $s \in [N]$ in $G$ with level $\ge i + 10$ via an $i$-short path, then it can reach a black node in $G'$ via an $i$-short path (including the case where $a$ itself was colored black).
  \end{enumerate}
  The second property is because we have colored all neighbors of the deleted nodes in black, thus any valid $i$-short path in $G$ will still lead to a black node in $G'$.

  Next, we count the number of vertices that can be reached from a black node via an $i$-short path. Since every node in $G'$ is low-degree, fixing a black node as the starting node, an $i$-short path of length $p > 0$ can be encoded with the following parameters:
  \begin{gather*}
    \l_1, \ldots, \l_p \in [0, L - 10] \cap \Z, \qquad t_k \in \Bk[\big]{2^{2^{L + \l_k + 5}}} \ \textup{ for } \ k \in [p], \\
    \textup{s.t.} \qquad \BK{k \in [p] \mid \l_k = i - \delta} \le \sqrt{2}^{\delta} \ \textup{ for all } \ 0 \le \delta \le i.
  \end{gather*}
  Given these parameters, the $i$-short path is determined edge-by-edge: $\l_k$ denotes the level of the $k$-th edge on the path; $t_k$ indicates one of the associated level-$\l_k$ edges of the $k$-th node on the path. We upper bound the number $C_1$ of different configurations for $(\l_1, \ldots, \l_p)$ by enumerating $q_\delta \defeq \abs{\BK{k \in [p] \mid \l_k = i - \delta}} \in \Bk[\big]{0, \sqrt{2}^\delta} \cap \Z$:
  \newcommand{\maxq}[1]{\floor{\sqrt{2}^{#1}}}
  \begin{align*}
    C_1 &\le \sum_{q_0 = 0}^{\maxq{0}} \sum_{q_1 = 0}^{\maxq{1}} \sum_{q_2 = 0}^{\maxq{2}} \cdots \sum_{q_i = 0}^{\maxq{i}} \binom{q_0 + q_1 + \cdots + q_i}{q_0, \, q_1, \, \ldots \, , \, q_i} \\
        &\le \prod_{\delta = 0}^i \bk[\big]{\sqrt{2}^\delta + 1} \cdot \binom{\maxq{0} + \maxq{1} + \cdots + \maxq{i}}{\maxq{0}, \, \maxq{1}, \,  \ldots \, , \maxq{i}} \\
        &\le 2^{O(L^2)} \cdot \prod_{j = 1}^i \binom{\maxq{0} + \cdots + \maxq{j}}{\maxq{j}} \\
        &\le 2^{O(L^2)} \cdot \prod_{j = 1}^i 2^{O(\sqrt{2}^j)} \\
        &\le 2^{O(L^2)} \cdot 2^{O(2^{L/2})} = 2^{o(2^L)}.
  \end{align*}
  Given any configuration of $(\l_1, \ldots, \l_p)$, the number of configurations for $(t_1, \ldots, t_p)$ is bounded by
  \begin{align*}
    C_2 &\le \prod_{\delta = 0}^i \bk[\big]{2^{2^{L + i - \delta + 5}}}^{\sqrt{2}^\delta} = 2^{(2^5 \cdot \sum_{\delta = 0}^i 2^{L + i - \delta} \cdot 2^{\delta / 2})} = 2^{(32 \cdot \sum_{\delta = 0}^i 2^{L + i - \delta / 2})} \le 2^{(32 \cdot \frac{1}{1 - 1 / \sqrt{2}})(2^{L + i})} \ll 2^{110 \cdot 2^{L + i}}.
  \end{align*}
  Therefore, the number of reachable nodes from any given black node is at most
  \[
    C_1 \cdot C_2 \le 2^{(110 + o(1)) \cdot 2^{L + i}} \ll 2^{111 \cdot 2^{L + i}}. \numberthis \label{eq:num-reachable}
  \]

  Lastly, we bound the number of black nodes, which consist of three parts:
  \begin{enumerate}
  \item The neighbors of high-degree nodes. According to \cref{lem:right-degree}, with probability $1 - 1 / \poly L$, the number of these black nodes does not exceed $n / 2^{\Omega(2^{2^{L}})}$.
  \item The non-empty slots with level $\ge i + 10$. Since the expected average probe complexity over all non-empty slots is at most $O\bk[\big]{2^{\sqrt{\log \eps^{-1}}}} = O(2^{2^L})$, with probability $1 - 1 / \poly L$, the total probe complexity is at most $O\bk[\big]{n \cdot 2^{2^L} \cdot \poly L}$; every non-empty slot with level $\ge i + 10$ has probe complexity at least $2^{2^{L + i + 9}}$, thus the number of such slots does not exceed
    \[
      \frac{O\bk[\big]{n \cdot 2^{2^L} \cdot \poly L}}{2^{2^{L + i + 9}}} \ll \frac{n}{2^{510 \cdot 2^{L + i}}}.
    \]
  \item The empty slots (with level $L$). There are at most
    \[
      O(\eps n) \le O(n / 2^{2^{2L - 1}}) \le O(n / 2^{2^{L + i + 9}}) \ll \frac{n}{2^{511 \cdot 2^{L + i}}}
    \]
    empty slots as the load factor equals $1 - \eps$.
  \end{enumerate}
  Adding them together, we know that the number of black nodes is at most $O(n / 2^{510 \cdot 2^{L + i}})$ with probability $1 - 1 / \poly L$. Multiplying with \eqref{eq:num-reachable}, we conclude that with probability $1 - 1 / \poly L$, only $n / 2^{399 \cdot 2^{L + i}} = n / 2^{\Omega(2^{L + i})}$ nodes can be reached via an $i$-short path from a black node. Finally, since $a_1 \in A_1$ and $a_2 \in A_2$ are sampled randomly, they do not belong to the reachable nodes with probability $1 - 1 / \poly L$, thus the lemma holds.
\end{proof}

Based on \cref{lem:path}, we can now analyze the desired properties of $\Phi$.

\begin{lemma}[Property 1 of $\Phi$]
  \label{lem:prop-1}
  Any insertion or deletion increases $\Phi$ by at most $1 / \poly L$ in expectation.
\end{lemma}

\begin{proof}
  Suppose we insert a random key $a \in A_2$ to the current hash table. For each $0 \le i \le L - 10$, $\Phi_{i + 10}$ is increased by at most 1. A necessary condition for $\Phi_{i + 10}$ to be increased is that there exists an $(i + 10)$-stanza, starting with the special slot $s_1$ containing the inserted key $a$, with positive potential. Say the stanza consists of slots $s_1, \ldots, s_j$ with $x_1, \ldots, x_{j - 1}$ residing in slots $s_1, \ldots, s_{j - 1}$ respectively, where $s_1$ is the special slot and $x_1 = a$. In the key-slot graph $G$ of the state before the insertion, this implies an $i$-short path from $a = x_1$ to $s_j$: $x_1 - s_2 - x_2 - \cdots - s_{j - 1} - x_{j - 1} - s_j$. The last vertex on the path is a slot with level $\ge i + 10$. Due to \cref{lem:path}, the probability of existing such a path is at most $1 / \poly L$, i.e., $\Phi_i$ increases by $1 / \poly L$ in expectation. Taking summation over all $i$, we know that $\Phi$ increases by $1 / \poly L$ in expectation as well.

  Similarly, suppose we delete a random key $a \in A_1$ from the current hash table. For each $0 \le i \le L - 10$, in order to increase $\Phi_{i + 10}$ by at most 1, there must be an $(i + 10)$-stanza $s_1, \ldots, s_j$, with $x_k$ stored in $s_k$ for $k \in [j - 1]$, where $s_j$ is the slot containing key $a$ before the deletion ($s_j$ becomes empty after the deletion and thus can serve as the final slot). In the key-slot graph $G$ of the state before the deletion, the stanza corresponds to an $i$-short path $s_1 - x_2 - s_2 - x_3 - \cdots - x_{j - 1} - s_j - a$. It is connecting the key $a$ to delete with a slot $s_1$ with level $\ge i + 10$, thus by \cref{lem:path}, such a path exists with probability at most $1 / \poly L$. Taking summation over $0 \le i \le L - 10$, we conclude that each deletion increases $\Phi$ by at most $1 / \poly L$, and the lemma follows.
\end{proof}

\begin{lemma}[Property 2 of $\Phi$]
  \label{lem:prop-2}
  If the algorithm performs a move with impact $r$, $\Phi$ will be decreased by $r \pm O(1)$.
\end{lemma}

\begin{proof}
  Without loss of generality, we assume the algorithm can only (a) move a key from a non-special slot to the special slot, and (b) move a key from the special slot to a non-special slot. This is because any move between non-special slots can be replaced by two moves of types (a) and (b) respectively. By symmetry, it suffices to analyze the moves of type (b).

  Suppose the algorithm moves a key $x$ from the special slot $s_1$ to a non-special slot $s_2$, where $\l(x, s_2) = j$. The move has impact $r = L - j$. For the sake of discussion, we denote by $\Sigma$ the state of the hash table before the move, and $\Sigma'$ the state after the move. Let $\Phi$ (resp.~$\Phi'$) be the potential of $\Sigma$ (resp.~$\Sigma'$), and let $\Phi_i$ (resp.~$\Phi'_i$) be the $i$-th summation term in $\Phi$ (resp.~$\Phi'$).

  We will show that for each $10 \le i \le L$:
  \begin{enumerate}
  \item\label{case:1} if $i \le j$, then $\Phi'_i = \Phi_i$;
  \item\label{case:2} if $j < i < j + 10$, then $\Phi'_i - \Phi_i \in [-2, \, 0]$;
  \item\label{case:3} if $i \ge j + 10$, then $\Phi'_i - \Phi_i \in \Bk[\big]{-1 - O(\sqrt{2}^{j-i}), \, -1 + O(\sqrt{2}^{j-i})}$.
  \end{enumerate}

  \textbf{Case \ref{case:1}.} If we represent every stanza $(\tilde{s}_1, \tilde{s}_2, \ldots, \tilde{s}_j)$ using the tuple $(x_1, \tilde{s}_2, \tilde{s}_3, \ldots, \tilde{s}_j)$, where $x_1$ is the key residing in $\tilde{s}_1$, then the set of representations of valid $i$-stanzas in $\Sigma$ and $\Sigma'$ are the same, also with the same potentials. Therefore, $\Phi'_i = \Phi_i$ holds.

  \smallskip

  \textbf{Case \ref{case:2}.} The only changes from $\Sigma$ to $\Sigma'$ are:
  \begin{itemize}
  \item $s_1$ is no longer a valid starting slot;
  \item $s_2$ is no longer a valid final slot.
  \end{itemize}
  Every valid $i$-stanza in $\Sigma'$ is also valid in $\Sigma$ with the same potential, so $\Phi'_i \le \Phi_i$; each of the above two changes will invalidate at most 1 stanza from the collection of disjoint $i$-stanzas in $\Sigma$, each decreasing the potential by at most 1, so $\Phi'_i \ge \Phi_i - 2$. So the statement also holds.

  \smallskip

  \textbf{Case \ref{case:3}.} Let $C$ be a collection of disjoint $i$-stanzas in $\Sigma$ with the maximum (sum of) potential. Let $s_1 \circ c_1$ be the stanza in $C$ that starts with the special slot $s_1$, where $c_1$ is a sequence of slots, if such a stanza exists; let $c_2 \circ s_2$ be the stanza in $C$ that ends with $s_2$, where $c_2$ is a sequence of slots, if such a stanza exists.

  We first adjust $C$ to make sure that both stanzas $s_1 \circ c_1$ and $c_2 \circ s_2$ exist while decreasing the potential of $C$ by at most $O(\sqrt{2}^{j - i})$. If $s_1 \circ c_1$ does not exist, we remove any stanza ending with $s_2$, decreasing the potential by at most $1$; then add the stanza $(s_1, s_2)$ with potential $1 - \sqrt{2}^{j - i + 10}$. Similarly, if $c_2 \circ s_2$ does not exist, we remove any stanza starting with $s_1$ and add the stanza $(s_1, s_2)$. Further, if $s_1 \circ c_1$ and $c_2 \circ s_2$ are the same stanza, then we replace it with stanza $(s_1, s_2)$, which can only decrease the potential by at most $\sqrt{2}^{j - i + 10}$ as well. Below, we assume $C$ is the adjusted collection of disjoint stanzas with potential at least $\Phi_i - O(\sqrt{2}^{j - i})$.

  Then, we show $\Phi'_i \ge \Phi_i - 1 - O(\sqrt{2}^{j - i})$ by constructing a set $C'$ of disjoint stanzas in $\Sigma'$. If $(s_1, s_2) \in C$, $C' = C \setminus \BK{(s_1, s_2)}$ has potential at least $\Phi_i - 1$. Otherwise, $C$ must contain two different stanzas $s_1 \circ c_1$ and $c_2 \circ s_2$. We let $C' = C \setminus \BK{s_1 \circ c_1, c_2 \circ s_2} \cup \BK{c_2 \circ s_2 \circ c_1}$ and show the following facts:
  \begin{itemize}
  \item The new stanza $c_2 \circ s_2 \circ c_1$ is a valid $i$-stanza in $\Sigma'$, because:
    \begin{itemize}
    \item The first slot in $c_2$ and the last slot in $c_1$ have levels $\ge i$, as they serve as starting and final slots of stanzas in $\Sigma$, and their accommodated keys are unchanged in $\Sigma'$.
    \item The internal slots of the new stanza have levels at most $i - 10$. For $s_2$, this is because $\l(x, s_2) = j \le i - 10$; for other slots, it is because they are internal slots of the two original stanzas.
    \item Every key in the new stanza (except the key residing in the last slot of $c_1$, if there is one) can be stored in the subsequent slot with level at most $i - 10$. For key $x$ in slot $s_2$, this is because the stanza $s_1 \circ c_1$ requires that $\l(x, \, \textup{first slot of $c_1$}) \le i - 10$; for other keys, it is a requirement of stanzas $s_1 \circ c_1$ and $c_2 \circ s_2$ in $\Sigma'$.
    \end{itemize}
  \item The potential of the new stanza satisfies $\phi(c_2 \circ s_2 \circ c_1) = \phi(s_1 \circ c_1) + \phi(c_2 \circ s_2) - 1 - O(\sqrt{2}^{j - i})$, because if we compare the summations in \cref{def:potential-stanza} for the three stanzas, there will be only one extra term for $c_2 \circ s_2 \circ c_1$: $\sqrt{2}^{\l(x, s_2) - i + 10} = \sqrt{2}^{j - i + 10} = O(\sqrt{2}^{j - i})$.
  \end{itemize}
  These facts imply that $C'$ is a valid set of disjoint stanzas in $\Sigma'$ with potential at least $\Phi_i - 1 - O(\sqrt{2}^{j - i})$.

  Finally, we show that $\Phi_i \ge \Phi'_i + 1 - O(\sqrt{2}^{j - i})$. Similar to above, we let $\tilde{C}'$ be the collection of $i$-stanzas in $\Sigma'$ with the maximum sum of potential $\Phi'$. If $\tilde C'$ does not contain a stanza using slot $s_2$, we simply set $\tilde C = \tilde C' \cup \BK{(s_1, s_2)}$ to be a collection of disjoint $i$-stanzas in $\Sigma$, with potential $\Phi'_i + 1 - O(\sqrt{2}^{j - i})$. Otherwise, we suppose there is a stanza $c_2 \circ s_2 \circ c_1$ using $s_2$ in $\tilde C'$, as $s_2$ can only serve as an internal slot. We let $\tilde C = \tilde C' \setminus \BK{c_2 \circ s_2 \circ c_1} \cup \BK{s_1 \circ c_1, c_2 \circ s_2}$, and by a similar reasoning as above, we know the potential of $\tilde C$ is at least $\Phi'_i + 1 - O(\sqrt{2}^{j - i})$.

  \smallskip
  Putting all three cases together, we take a summation of $\Phi'_i - \Phi_i$ over $10 \le i \le L$, which gives $\Phi' - \Phi = -(L - j) \pm O(1) = -r \pm O(1)$ and concludes the proof.
\end{proof}

Lastly, Property 3 of $\Phi$ is true, because each $\Phi_i$ corresponds to a collection of disjoint stanzas, whose size is at most $O(n)$, while each stanza has potential at most 1, thus $\Phi_i \le O(n)$; taking a summation over $i$ shows $\Phi \le O(nL)$.

So far, we have proved all 3 desired properties of the potential function $\Phi$. By \cref{prop:prop-imply-thm}, these properties imply \cref{thm:lowerbound/uniform}.

\subsection{Non-nearly-uniform probe sequences}

The previous subsection proves the lower bound on the expected switching cost per operation based on the assumption of near uniformity. We recall that the probe-sequence function $h$ is \defn{nearly uniform} if
\[
  q(h, i, s) \defeq n \Pr_{x \in [U]} [\ProbeComplexity(h, x, s) \le i] \le O(i^{10}) \numberthis \label{eq:def-q}
\]
for all $i$ and $s$. For clarity, we use $\ProbeComplexity(h, x, s)$ to denote the probe complexity of key $x$ storing in slot $s$ under the probe-sequence function $h$. In this subsection, we reduce every probe-sequence function $h$ to a nearly uniform one $h'$, changing the expected average probe complexity by only a constant factor. This will remove the assumption of near uniformity in \cref{thm:lowerbound/uniform}.

Formally, let $A$ be an assignment of $n$ keys to $N$ slots. We denote by $c(A, h)$ the total probe complexity of all $n$ balls under the probe-sequence function $h$.

\begin{lemma}
  \label{lem:non-uniform}
  For every (deterministic) probe-sequence function $h$, we can construct a nearly uniform probe-sequence function $h'$, such that for any assignment $A$ of $n$ keys to $N$ slots, $c(A, h') \le O(c(A, h) + n)$.
\end{lemma}

By definition, the probe sequence $h(x) = (h_1(x), h_2(x), \ldots)$ contains at most one slot $h_i(x)$ on its $i$-th position. However, for the sake of discussion, we introduce \defn{generalized probe sequences} $(\tilde{h}_1(x), \tilde{h}_2(x), \ldots)$, where on the $i$-th position of the sequence there is a set $\tilde{h}_i(x)$ of slots. It is required that, among the first $i$ positions on the generalized probe sequence, the number of slots should not exceed $i$.

For a generalized-probe-sequence function $\tilde h$, $\ProbeComplexity(\tilde{h}, x, s)$ is still defined as the first position $i$ containing slot $s$; $q(\tilde{h}, i, s)$ and the notion of near uniformity are still defined according to \eqref{eq:def-q}. The proof consists of two steps: first, we construct a nearly uniform generalized-probe-sequence function $\tilde{h}$ that meets the requirements; second, we show that any (nearly uniform) generalized probe sequence can be transformed into a (nearly uniform) probe sequence with little overhead.

\paragraph{Step 1: Constructing generalized probe sequences.}

We construct $\tilde{h}$ by reassigning some occurrences of slots in the probe sequences to later positions. For each pair $(i, s)$ where $q(h, i, s) > i^{5}$, which we call a \defn{bad pair}, and for each key $x$ where $h_i(x) = s$, we move the occurrence of $s$ from $h_i(x)$ to a later position $h_{\ceil[\big]{\!\sqrt{q(h, i, s)}}}(x)$, resulting in generalized probe sequences $\tilde{h}$.

Next, we show that $\tilde{h}$ meets our requirements.

\begin{claim}[Probe complexity of $\tilde{h}$]
  \label{claim:complexity}
  For any assignment $A$ of $n$ keys to $N$ slots, the total probe complexity of $\tilde h$ is at most $c(A, \tilde{h}) \le c(A, h) + O(n)$.
\end{claim}

\begin{proof}
  For each integer $Q$ that is a power of two, the number of bad pairs $(i, s)$ where $q(h, i, s) \in [Q, 2Q)$ is bounded by
  \begin{align*}
    &\phantom{{}=} \sum_{i=1}^{(2Q)^{1/5}} \abs*{\myset*{s \in [N]}{\Pr_{x \in [U]}\Bk[\big]{\ProbeComplexity(h, x, s) \le i} = \Theta(Q / n)}} \\
    &\le \sum_{i=1}^{(2Q)^{1/5}} O(i \cdot n / Q) \\
    &= O\bk[\big]{n / Q^{3/5}},
  \end{align*}
  where the first equality is because $\sum_{s \in [N]} \Pr_{x \in [U]}[\ProbeComplexity(h, x, s) \le i]$ does not exceed $i$. For each of the bad pairs, moving the occurrence of $s$ to position $\ceil*{q(h, i, s)^{1/2}} = \Theta(Q^{1/2})$ may increase the probe complexity of at most one key by $O(Q^{1/2})$ (potentially, the key $x$ stored in slot $s$ in the assignment $A$ will get a higher probe complexity). Taking a summation over all $Q$, we upper bound the increment on the total probe complexity by
  \[
    c(A, \tilde{h}) - c(A, h) \le \sum_{Q \ge 1 \textup{ is power of two}} O\bk[\big]{n / Q^{3/5}} \cdot O\bk[\big]{Q^{1/2}} = O(n). \qedhere
  \]
\end{proof}

\begin{claim}
  \label{claim:uniform}
  $\tilde h$ is nearly uniform.
\end{claim}

\begin{proof}
  For each $i \ge 1$ and slot $s$, we have
  \begin{align*}
    & \Pr_{x \in [U]} \Bk[\big]{\ProbeComplexity(\tilde{h}, x, s) = i} \\
    ={} & \Pr_{x \in [U]} \Bk[\big]{\ProbeComplexity(h, x, s) = i \textup{ and $(i, s)$ is not a bad pair}} \numberthis \label{eq:term-old} \\
    +{} & \Pr_{x \in [U]} \Bk[\big]{\ProbeComplexity(\tilde{h}, x, s) = i \textup{ and $s \in \tilde{h}_i(x)$ is a newly-assigned element in $\tilde{h}$}}. \numberthis \label{eq:term-new}
  \end{align*}
  \cref{eq:term-old} is at most $i^{5} / n$ since $(i, s)$ is not bad. When the event in \eqref{eq:term-new} occurs, there must be a position $i' < i$ where $q(h, i', s) = \Theta(i^2)$, i.e., $\Pr_{x \in [U]}[\ProbeComplexity(h, x, s) \le i'] = \Theta(i^2 / n)$. Thus, we can bound \eqref{eq:term-new} by enumerating $i' < i$:
  \begin{align*}
    \textup{\cref{eq:term-new}} \;
    ={} & \sum_{i' < i} \Pr_{x \in [U]}\Bk[\big]{\ProbeComplexity(h, x, s) = i' \textup{ and } s = h_{i'}(x) \textup{ is reassigned to } \tilde{h}_{i}(x)} \\
    \le{} & \sum_{i' < i} O(i^2 / n) = O(i^3 / n).
  \end{align*}
  Adding \eqref{eq:term-old} and \eqref{eq:term-new} together, and taking a summation over $i \in [i_0]$, we get
  \[
    \Pr_{x \in [U]} \Bk[\big]{\ProbeComplexity(\tilde{h}, x, s) \le i_0} \le O(i_0^6 / n),
  \]
  so $\tilde h$ is nearly uniform.
\end{proof}

\paragraph{Step 2: Transforming into (normal) probe sequences.}

We have already constructed a generalized-probe-sequence function $\tilde h$ that has low probe complexity and is nearly uniform. The last step is to transform it into a probe-sequence function.

\begin{claim}
  \label{claim:transform}
  Let $\tilde h$ be a nearly uniform generalized-probe-sequence function. There exists a nearly uniform probe-sequence function $h'$, such that for any assignment $A$ of $n$ keys to $N$ slots, the total probe complexity $c(A, h') \le 2c(A, \tilde{h})$.
\end{claim}

\begin{proof}
  The given function $\tilde h$ maps each key $x$ to a generalized probe sequence $(\tilde h_1(x), \tilde h_2(x), \ldots)$, where each position accommodates a set of slots (possibly empty). For each $\tilde h_i(x) = \varnothing$, we insert a single element $\Null$ into it. After that, we write down all elements in the order they appear in the generalized probe sequence (when a position $\tilde h_i(x)$ contains multiple slots, we write in an arbitrary order), forming a probe sequence consisting of $[N] \cup \BK{\Null}$. We denote it by $h'$.

  For any slot $s$ that occur in $\tilde h_{i}(x)$, its position in $h'$ after the above transformation will be between $i$ and $2i$ (both included): it is at least $i$ because we inserted $\Null$s to make sure every position contains at least one element; it is at most $2i$ because (1) there can only be $i$ elements among the first $i$ positions, before we insert any $\Null$; (2) the number of $\Null$s we inserted to the first $i$ positions is at most $i$. This implies that
  \[
    \ProbeComplexity(\tilde{h}, x, s) \le \ProbeComplexity(h', x, s) \le 2 \cdot \ProbeComplexity(\tilde{h}, x, s)
  \]
  for every key $x$ and slot $s$. So we conclude that (1) $h'$ is nearly uniform; (2) $c(A, h') \le 2c(A, \tilde{h})$ for any assignment $A$ of keys.
\end{proof}

Combining \cref{claim:complexity,claim:uniform,claim:transform} together, we have proved \cref{lem:non-uniform}.

\subsection{Putting pieces together}

Combining the results from the previous subsections, we can derive a formal lower bound on classical open-addressing.

Specifically, for any classical open-addressing hash table, we analyze its performance on \cref{dist:hard}. By Yao's minimax principle, we only need to consider deterministic hash tables, i.e., the probe-sequence function $h$ is fixed. Then, \cref{lem:non-uniform} transforms $h$ to a nearly uniform function $h'$ while only losing a constant factor on the average probe complexity on any state of the hash table. Finally, \cref{thm:lowerbound/uniform} shows a lower bound on the switching cost per operation, provided that the expected average probe complexity is small enough. We summarize the result as the following theorem.

\thmlower*

\section{Open Problems}

We conclude the paper with several appealing open questions.

\paragraph{Non-oblivious open addressing.} The first question concerns \defn{non-oblivious open addressing} \cite{fiat1988nonoblivious,fiat1993implicit}: this is a generalization of open-addressing in which queries are not constrained to follow any particular probe sequence. Instead, insertions/queries/deletions can be implemented arbitrarily subject to the constraint that the \emph{state} of the data structure, at any given moment, is that of an open-addressed hash table. Formally, this means that, if the hash table is storing $n$ keys, then its state is an array with $n$ non-empty slots, where the non-empty slots contain some permutation of the keys being stored. 

All classical (a.k.a.~oblivious) open-addressed hash tables are also valid non-oblivious open-addressed hash tables. Thus the upper bounds in this paper also apply to the non-oblivious case. However, the \emph{lower bounds} do not. This raises the following question: can a non-oblivious open-addressed hash table hope to achieve $O(1)$ expected-time queries while also achieving an expected insertion/deletion time of $o(\log \log \epsilon^{-1})$?

\paragraph{High-probability worst-case query time bounds.} One major direction in recent decades has been to develop open-addressed hash tables that support high load factors while also offering (high-probability) worst-case query time bounds. Using variations of Cuckoo hashing \cite{dietzfelbinger2007balanced,bell20241,fotakis2005space}, one can achieve $O(\log \epsilon^{-1})$ worst-case query time, while also supporting $f(\epsilon^{-1})$ expected insertion/deletion time for some function $f$.  One can also show using coupon-collector-style arguments that this $O(\log \epsilon^{-1})$ bound is the best (worst-case) bound that one can hope for. What is not clear is whether one might also be able to ask for a very good insertion/deletion time. Can one achieve $O(\log \epsilon^{-1})$ worst-case queries (w.h.p.) while also supporting $o(\epsilon^{-1})$ expected insertion/deletion time? Or, more generally, can one hope to achieve (high-probability) worst-case query time $Q$ and expected insertion/deletion time $I$ for some $Q$ and $I$ satisfying $QI = o(\epsilon^{-1})$? We conjecture that such a bound should not be possible.

\paragraph{High-probability worst-case time bounds for all operations.} Finally, it is also interesting to consider the task of achieving high-probability worst-case time bounds for \emph{all} operations (insertions, deletions, and queries). For example, if we consider a load factor of $1/2$, what are the best (worst-case) bounds that a classical open-addressed hash table can hope to achieve (as a function of $n$)? Using results from the power of two choices \cite{vocking2003how,dalal2023twoway}, one can achieve a bound of $O(\log \log n)$. Is this the best bound possible?

\bibliographystyle{plain}
\bibliography{bib}

\end{document}